\newif\ifFullVersion
\FullVersiontrue

\documentclass[a4paper,USenglish,numberwithinsect]{lipics-v2019}

\ifFullVersion
\hideLIPIcs
\nolinenumbers
\fi

\usepackage{cite} \usepackage{microtype}

\usepackage{comment}
\usepackage{xspace}
\usepackage{mathtools}
\usepackage{algorithm,algorithmic}

\usepackage{tikz}
\usetikzlibrary{calc,intersections,through,backgrounds}
\usetikzlibrary{arrows, decorations.pathreplacing}
\usepackage{pgfplots}
\newcommand{\AxisRotator}[1][rotate=0]{\tikz [x=0.25cm,y=0.60cm,line width=.2ex,-stealth,#1] \draw (0,0) arc (-150:150:1 and 1);}

\newcommand{\BeginMyItemize}{\begin{itemize}\setlength{\itemsep}{-\parskip}}
\newcommand{\EndMyItemize}{\end{itemize}}

\newcommand{\BeginMyEnumerate}{\begin{enumerate}\setlength{\itemsep}{-\parskip}}
\newcommand{\EndMyEnumerate}{\end{enumerate}}
\newcommand{\myenumerate}[1]{\BeginMyEnumerate #1 \EndMyEnumerate}

\renewcommand{\leq}{\leqslant}
\renewcommand{\geq}{\geqslant}
\newcommand{\mypara}[1]{\vspace{10pt} \noindent \textbf{\sffamily #1}}

\theoremstyle{plain}

\newenvironment{myquote}{\list{}{\leftmargin=4mm\rightmargin=4mm}\item[]}{\endlist}
\newenvironment{claiminproof}{\begin{myquote}\noindent\emph{Claim.}}{\end{myquote}}
\newenvironment{proofinproof}{\begin{myquote}\noindent\emph{Proof.}}{\hfill $\lhd$ \end{myquote}}

\newenvironment{proofof}[1]{

\noindent{\textbf{\sffamily Proof of {#1}.}}}
{\hfill\qed

}

\newcommand{\Reals}{{\mathbb R}}
\newcommand{\Nats}{{\mathbb N}}

\DeclareMathOperator{\interior}{Int}

\DeclareMathOperator{\vol}{Vol}

\DeclareMathOperator{\radius}{radius}
\DeclareMathOperator{\dist}{dist}

\newcommand{\node}{\nu}

\newcommand{\graph}{\ensuremath{\mathcal{G}}}

\newcommand{\eps}{\varepsilon}
 \newcommand{\etal}{\emph{et al.}\xspace}

\newcommand{\NN}{\text{\sc nn}\xspace}
\newcommand{\CI}{\text{\sc ci}\xspace}
\newcommand{\optalg}{\text{\sc Opt}\xspace}
\newcommand{\objects}{P}

\newcommand{\alg}{\text{\sc alg}\xspace}
\newcommand{\family}{\mathcal{F}\xspace}

\newcommand{\sol}{\mathcal{S}}
\newcommand{\set}{S}
\newcommand{\radii}{\mathcal{R}}

\DeclareMathOperator{\OPT}{OPT}
\DeclareMathOperator{\opt}{\OPT}
\DeclareMathOperator{\nn}{nn}
\DeclareMathOperator{\cost}{cost}
\DeclareMathOperator*{\argmin}{arg\,min}
\DeclareMathOperator*{\argmax}{arg\,max}
\DeclareMathOperator{\parb}{Par}

\definecolor{mybrown}{HTML}{65000b} \definecolor{myred}{RGB}{148,26,33} \definecolor{myorange}{HTML}{d56638} \definecolor{mybeige}{HTML}{f8b878} \definecolor{mygreen}{HTML}{014421} 

\definecolor{mybrown-line}{HTML}{7f000e} \definecolor{myred-line}{HTML}{db2932} \definecolor{myorange-line}{HTML}{ef733e} \definecolor{mybeige-line}{HTML}{ffb368} \definecolor{mygreen-line}{HTML}{006d34} 

\definecolor{mygray-light-color}{HTML}{c3baba} \definecolor{mygray-medium-color}{HTML}{878486} \definecolor{mygray-dark-color}{HTML}{394044} 

\definecolor{mygray-light}{gray}{0.75} \definecolor{mygray-medium}{gray}{0.6} \definecolor{mygray-dark}{gray}{0.4}

\title{The Online Broadcast Range-Assignment Problem}

\author{Mark de Berg}{
           Department of Computing Science, TU Eindhoven,
            the Netherlands}{m.t.d.berg@tue.nl}{https://orcid.org/0000-0001-5770-3784}{Supported by the Netherlands' Organisation for     Scientific Research (NWO) under project no.~024.002.003.}
\author{Aleksandar Markovic}{
           Department of Computing Science, TU Eindhoven,  the Netherlands}{info@aleksmarkovic.com}{}{Supported by the Netherlands'  Organisation for Scientific Research (NWO) under project no.~024.002.003.}
\author{Seeun William Umboh}{
      School of Computer Science, The University of Sydney, Australia}{william.umboh@sydney.edu.au} {https://orcid.org/0000-0001-6984-4007}{Supported by NWO grant 639.022.211.}

\authorrunning{M. de Berg, A. Markovic, and S. W. Umboh} 

\Copyright{Mark de Berg, Aleksandar Markovic, and William Umboh}

\ccsdesc[100]{Theory of computation~Design and analysis of algorithms}
\keywords{Computational geometry, online algorithms, range assignment, broadcast}

\EventEditors{Yixin Cao, Siu-Wing Cheng, and Minming Li}
\EventNoEds{3}
\EventLongTitle{31st International Symposium on Algorithms and Computation (ISAAC 2020)}
\EventShortTitle{ISAAC 2020}
\EventAcronym{ISAAC}
\EventYear{2020}
\EventDate{December 14--18, 2020}
\EventLocation{Hong Kong, China (Virtual Conference)}
\EventLogo{}
\SeriesVolume{181}
\ArticleNo{8}

\begin{document}

\maketitle

\begin{abstract}
Let $\objects=\{p_0,\ldots,p_{n-1}\}$ be a set of points in $\Reals^d$, modeling devices 
in a wireless network. A range assignment assigns a range $r(p_i)$ to each 
point $p_i\in\objects$, thus inducing a directed communication graph $\graph_r$ in which 
there is a directed edge $(p_i,p_j)$ iff $\dist(p_i, p_j) \leq r(p_i)$,
where $\dist(p_i,p_j)$ denotes the distance between~$p_i$ and~$p_j$.
The range-assignment problem is to assign the transmission ranges such
that $\graph_r$ has a certain desirable property, while minimizing the cost of the
assignment; here the cost is given by $\sum_{p_i\in\objects} r(p_i)^{\alpha}$,
for some constant~$\alpha>1$ called the distance-power gradient.

We introduce the online version of the range-assignment problem, where the points
$p_j$ arrive one by one, and the range assignment has to be updated at each
arrival. Following the standard in online algorithms, resources given out cannot be
taken away---in our case this means that the transmission ranges will never decrease.
The property we want to maintain is that $\graph_r$
has a broadcast tree rooted at the first point~$p_0$. Our results include the following.
\begin{itemize}
    \item We prove that already in $\Reals^1$, a 1-competitive algorithm does not exist.
          In particular, for distance-power gradient~$\alpha=2$ any online algorithm has
          competitive ratio at least~1.57.
    \item For points  in $\Reals^1$ and $\Reals^2$, we analyze two natural strategies for updating the 
          range assignment upon the arrival of a new point $p_j$. The strategies
          do not change the assignment if $p_j$ is already within range of an
          existing point, otherwise they increase the range of a single point, as follows:
          {\sc Nearest-Neighbor} (\NN) increases the range of $\nn(p_j)$, the nearest
          neighbor of~$p_j$, to $\dist(p_j, \nn(p_j))$,  
          and {\sc Cheapest Increase} (\CI) increases the range of
          the point $p_i$ for which the resulting cost increase to be able to reach the new point~$p_j$ is minimal.
          We give lower and upper bounds on the competitive
          ratio of these strategies as a function of the distance-power gradient~$\alpha$.
          We also analyze the following variant of \NN in $\Reals^2$ for $\alpha=2$:
          {\sc $2$-Nearest-Neighbor} (2-\NN) increases the range of $\nn(p_j)$ to $2\cdot \dist(p_j,\nn(p_j))$,
     \item We generalize the problem to points in arbitrary metric spaces, where  
           we present an $O(\log n)$-competitive algorithm.
\end{itemize}

\end{abstract}

\section{Introduction}
Consider a collection of wireless devices, each with its own transmission range.
The transmission ranges induce a directed communication network, where each 
device $p_i$ can directly send a message to any device $p_j$ in its 
transmission range. If $p_j$ is not within range, a message from $p_i$ can still 
reach $p_j$ if there is a path from $p_i$ to $p_j$ in the communication network.
The energy consumption of a device depends on its transmission range: the greater
the range, the more power is needed. 
This leads to the range-assignment problem: assign transmissions ranges to the
devices such that the resulting network has some desired connectivity property,
while minimizing the total power consumption.

Mathematically we can model the problem as follows.  
Let~$\objects=\{p_0,\ldots,p_{n-1} \}$ be a set of~$n$ points in~$\Reals^d$. For an 
assignment~$r:\objects \to \Reals_{\geqslant 0}$, let~$\graph_r$ be the 
directed graph on the vertex set~$\objects$ obtained by putting 
a directed edge from a vertex~$p_i$ to a vertex~$p_j$ iff~$\dist(p_i,p_j) \leqslant r(p_i)$,
where $\dist(p_i,p_j)$ denotes the distance between~$p_i$ and~$p_j$.
We call~$\graph_r$ the \emph{communication graph} on~$\objects$ induced by the range assignment~$r$. 
The \emph{cost} of a range assignment~$r$ is defined as 
$\cost_\alpha (r):= \sum_{p_i\in \objects} r(p_i)^\alpha$,
where~$\alpha \geq 1$ is called the \emph{distance-power gradient}. 
In practice,~$\alpha$ typically varies from~$1$ to~$6$~\cite{Pahlavan_Levesque_book}.
We then want to find a range assignment that minimizes the cost while ensuring that
$\graph_r$ has some desired property. Properties that have been investigated
in this context include strong connectivity~\cite{cps-hrpra-99,kkkp-pcprn-00},  
$h$-hop strong connectivity~\cite{cpfps-mrap-03,cps-pap-04,kkkp-pcprn-00},
broadcast capability---here $\graph_r$ must contain a broadcast tree
(that is, an arborescence) rooted at the source point~$p_0$---,
and $h$-hop broadcast capability~\cite{acilmprs-ealebhb-04,gk-imanwst-99}; 
see the survey by Clementi~\etal~\cite{Survey_Clementi}
for an overview of the various range-assignment problems. Most previous work considered the Euclidean setting. There has been some work on arbitrary metric spaces for the strong connectivity version~\cite{grandoni12, calinescu13}. (Note that while the 2-dimensional version seems the most relevant setting, the distances may
not be Euclidean due to obstacles that reduce the strength of the signal of a device.)

In this paper we focus on the broadcast version of the range-assignment problem.
This version can be solved optimally in a trivial manner when~$\alpha=1$,
by setting $r(p_0) := \max_{0\leqslant i < n} \dist(p_0,p_i)$ and $r(p_i):=0$ for~$i>0$.
Clementi~\etal~\cite{Clementi_range_assignment_1D} showed a polynomial 
time algorithm for the~$1$-dimensional problem when~$\alpha\geqslant 2$. 
Moreover, Clementi~\etal~\cite{Clementi_range_assiugnment_complexity} showed 
the problem is NP-hard for any~$\alpha > 1$ and any~$d\geq 2$. 
Clementi~\etal~\cite{Clementi_range_assignment_1D}, 
Clementi~\etal~\cite{Clementi_range_assiugnment_complexity}, 
and Wan~\etal~\cite{Wan2002} also showed that the problem 
can be approximated within a factor~$c\cdot 2^\alpha$ for any~$\alpha \geqslant 2$ 
and for a certain constant~$c$. Furthermore, 
Clementi~\etal~\cite{Clementi_range_assiugnment_complexity} showed 
that for any~$d\geqslant 2$ and for any~$\alpha \geqslant d$, 
there is a function~$f : \Nats \times \Reals \to \Reals$ 
such that the problem can be approximated within a factor~$f(d,\alpha)$  
in the~$d$-dimensional Euclidean space. 
Fuchs~\cite{Fuchs_range_assigment_hardness} 
showed that for $d=2$, the problem remains NP-hard even for so-called 
well-spread instances for any~$\alpha >1$. In dimension $d\geq 3$, he also showed that the problem is NP-hard to approximate within a factor of~$51/50$ when~$\alpha > 1$; the result also holds for well-spread instances when~$\alpha > d$. 

\subparagraph{Our contribution.}
We study the online version of the broadcast range-assignment problem. Here the points
$p_0,p_1,\ldots,p_{n-1}$ come in one by one, and the goal is to maintain a range assignment~$r$
such that $\graph_r$ contains a broadcast tree on the currently inserted points,
rooted at the first point~$p_0$. Of course one can simply recompute the assignment
from scratch, but in online algorithms one requires that resources that
have been given out cannot be taken back. For the range assignment problem 
this means that we are not allowed to decrease the range of any point. 
In fact, our algorithms have the useful property that upon arrival of each point,
we change the current range assignment only minimally: either we do not change it at
all---this happens when the newly arrived point is already within range of an existing point---or
we increase the range of only a single point. Our goal is to obtain algorithms with
a good competitive ratio.\footnote{The \emph{competitive ratio}~\cite{Sleator_Tarjan-1985} of an online algorithm
compares the cost of the solution it maintains to the cost achieved by the
optimal offline algorithm. More precisely, an online algorithm~\alg is~\emph{$c$-competitive} 
if there is a constant~$a$ such that for any instance~$I$, 
the cost of~\alg is at most~$c \cdot \opt(I) + a$. Note the the offline algorithm
must still maintain the solution over time. Thus the offline problem is not the same
as the static problem, where we only want a solution for the final configuration.}
As far as we know, the range-assignment problem
has not been studied from the perspective of online algorithms.
\medskip

We first prove a lower bound on the competitive ratio achievable by any online algorithm:
even in $\Reals^1$ there is a constant $c_{\alpha}>1$ (which depends
on the power-distance gradient~$\alpha$) such that no online algorithm can be $c_{\alpha}$-competitive.  
For $\alpha=2$, we have $c_{\alpha}>1.57$.

We then investigate the following two natural online algorithms for the broadcast range-assignment problem.
Suppose the point~$p_j$ arrives. Our algorithms all set $r(p_j) :=0$ and, as mentioned, they do not
change any of the ranges $r(p_0),\ldots,r(p_{j-1})$ if $|p_i p_j| \leq r(p_i)$ for some~$0\leq i<j$.
When $p_j$ is not within range of an already inserted point, the algorithms increase the range of 
one point, as follows. Let~$\nn(p_j)$ denote the nearest neighbor of~$p_j$ in the set~$\{ p_0,\ldots,p_{j-1} \}$,
with ties broken arbitrarily. 
\vspace*{2mm}
\begin{itemize}
\item {\sc Nearest-Neighbor} (\NN for short) increases the range of~$\nn(p_j)$ 
      to~$\dist(p_j,\nn(p_j))$.
\item {\sc Cheapest Increase} (\CI for short) increases the range of $p_{i^*}$ to $\dist(p_{i^*},p_j)$,
      where $p_{i^*}$ is a point minimizing the cost increase of the assignment, which 
      is~$\dist(p_{i^*},p_j)^{\alpha}-r(p_{i^*})^{\alpha}$
      where $r(p_{i^*})$ denotes the current range of~$p_{i^*}$.
\end{itemize}

\begin{table}[h]
\begin{tabular}{c|c|c|c}
 dimension & distance-power gradient & lower bound for \NN  & upper bound for \NN and \CI \\
\hline
$d=1$   &  $\alpha=2$ &      2          &   2   \\
\hline
\multirow{4}{*}{$d=2$}   &  \multirow{2}{*}{$\alpha=2$} &  $\approx 7.61$  &   322 \\
        &             &       &  2-\NN:  36      \\
\cline{2-4}
        &  $2<\alpha<\alpha^*\approx 4.3$ & \multirow{2}{*}{$\approx 6(1+0.52^\alpha)$} &  $\alpha \frac{2^\alpha -3}{2^{\alpha-1}-\alpha}$    \\
        &  $\alpha\geq\alpha^*\approx 4.3$ &   & $\approx 12.94$  \\
\end{tabular}
\caption{Overview of results on \NN and \CI.}
\label{table:results}
\end{table}
The results are summarized in Table~\ref{table:results}.
Note the lower bounds hold only for \NN, while the upper bounds hold for \NN and \CI;
the exception is the third row, which is for 2-\NN (see below). The lower bound of $6(1+0.52^\alpha)$
mentioned in the table---the exact bound is $6(1+(\frac{\sqrt{6}-\sqrt{2}}{2})^\alpha)$---applies to all $\alpha>1$, 
and thus implies the given lower bound for $\alpha=2$.
Recall that for $d=1$ and $\alpha=2$, we also have a universal lower bound of 1.57 that holds for any online algorithm
and, hence, also for \CI. The exact value of $\alpha^*$ is $\alpha^* = \argmin  F^*_\alpha$, where
$F^*_\alpha=\alpha \frac{2^\alpha -3}{2^{\alpha-1}-\alpha}$.

As can be seen in the table \NN is $O(1)$-competitive for $\alpha=2$, but the competitive ratio
is quite large, namely 322. We therefore also analyze the following variant of~$\NN$,
which (if $p_j$ is not yet within range of an existing point) proceeds as follows:
\vspace*{2mm}
\begin{itemize}
\item {\sc 2-Nearest-Neighbor} (2-\NN for short) increases the range of~$\nn(p_j)$ 
	  to $2\cdot \dist(p_j,\nn(p_j))$.
\end{itemize} 
\vspace*{2mm}
We prove that the competitive ratio of 2-\NN is at most~$36$ for~$\alpha=2$. Thus, while still rather large, 
the competitive ratio is a lot smaller than what we were able to prove for~\NN.
It is interesting to note that both \NN and~2-\NN make decisions that are
independent of~$\alpha$. Hence, \NN obtains a solution that is
simultaneously competitive for all $\alpha \geq 2$. 
\medskip

As a final contribution we generalize the broadcast problem to points in arbitrary metric spaces. 
Since to the best of our knowledge this version has not been studied before, 
we present 
an approximation algorithm for the offline setting; 
its approximation ratio is $5^{\alpha}$. In this offline setting the algorithm must
be what Boyar~\etal~\cite{befkl-ods-19} call an \emph{incremental algorithm}:
an algorithm that, even though it may know the future, maintains a feasible solution at any time.
For the online setting (where the future is unknown) we obtain an
$O(4^{\alpha}\log n)$-competitive algorithm. 

\subparagraph{Notation.}
We let $\objects := p_0,\ldots,p_{n-1}$ denote the input sequence, where we
assume without loss of generality that~$p_i$ is inserted at time~$i$ and that all $p_i$ are distinct.
Define~$\objects_i:= p_0,\ldots,p_i$, and 
denote the range of a point $p_i\in \objects_j$ just after the insertion
of the point $p_j$ by $r_j(p_i)$. Thus in the online version we require
that $r_j(p_i) \leqslant r_{j+1} (p_i)$. For an algorithm~\alg we use
$\cost_{\alpha}(\alg(\objects))$ to denote the cost incurred by \alg on input~$\objects$
for distance-power gradient~$\alpha$.
Finally we denote the ball of radius $\rho$ centered at a point $p$ by
$B(p,\rho)$; note that in $\Reals^1$ this is an interval of length~$2\rho$
and in $\Reals^2$ it is a disk of radius~$\rho$. \section{Online range-assignment in $\Reals^1$}
\label{se:1D}
In this section we prove that no online algorithm can have a competitive ratio 
arbitrarily close to~$1$, even in~$\Reals^1$. We also prove bounds on the competitive ratio of 
\NN and \CI in~$\Reals^1$. 

\mypara{A universal lower bound.}
To prove the lower bound we consider an arbitrary online algorithm~\alg.
Our adversary then first presents the points $p_0=0$, $p_1=x$, and $p_2:=\delta_\alpha \cdot x$. 
Depending on the range assignment \alg has done so far, the adversary either ends the instance
or presents a fourth point $p_3=-\delta_\alpha \cdot x$. By picking a suitable value
for $\delta_\alpha$ and making $x$ sufficiently large, we can obtain a lower bound.
This is made precise in the following theorem. 
\begin{theorem} \label{th:universal-lb}
For any distance-power gradient~$\alpha>1$, there is a constant~$c_\alpha>1$ such that any 
online algorithm for the range assignment problem in~$\Reals^1$ has a competitive ratio of at least~$c_\alpha$. 
For~$\alpha=2$ this constant is~$c_2 \approx 1.58$. 
\end{theorem}
\begin{proof}
Let~$\alpha > 1$ and let~\alg be an algorithm with competitive 
ratio~$c \geqslant 1$, i.e., there is a constant~$a$ such that the 
cost of~\alg is upper bounded by~$c\cdot \opt + a$. 
We also define 
\begin{align*}
c_\alpha :&= \max_{\delta>1} \min \left(
		\frac{\delta^\alpha}{1+(\delta-1)^\alpha}, 
		\frac{\delta^\alpha + (\delta-1)^\alpha}{\delta^\alpha}, 
		\frac{1 + (\delta+1)^\alpha}{\delta^\alpha}
		\right), \\
\mbox{and }\quad
\delta_\alpha :&= \argmax_{\delta>1}  \min \left(
		\frac{\delta^\alpha}{1+(\delta-1)^\alpha}, 
		\frac{\delta^\alpha + (\delta-1)^\alpha}{\delta^\alpha}, 
		\frac{1 + (\delta+1)^\alpha}{\delta^\alpha}
		\right).
\end{align*}
We show that~$c \geqslant c_\alpha$ by constructing 
the following families of instances consisting of, respectively, three and four points, 
and parametrized by the real number~$x\geqslant 1 $:
\begin{align*}
&\family_1:=\{ \{p_0=0, p_1=x, p_2=\delta_\alpha \cdot x \} \} \\
\mbox{and }\quad &\family_2:=\{ \{p_0=0,  p_1=x, p_2=\delta_\alpha \cdot x, p_3=-\delta_\alpha \cdot x \} \}. 
\end{align*}
Note that there is a one-to-one correspondence 
between the instances in both families: each instance of~$\family_1$ is 
the beginning of exactly one instance of~$\family_2$ and each instance of~$\family_2$ 
starts like exactly one instance of~$\family_1$. 

For any~$x$, depending on what~\alg 
does after~$p_2$ is inserted, we choose an instance from either 
the family~$\family_1$ or the family~$\family_2$ using the 
following rule: if after~$p_2$ is inserted,~\alg has a disk of radius 
at least~$\delta_\alpha \cdot x$, we choose~$\family_1$, otherwise we choose~$\family_2$. 
In the former case,~\alg pays at least~$\delta_\alpha^\alpha \cdot x^\alpha$ while the optimal solution 
would be to place a disk of radius~$x$ at~$p_0$ 
and a disk of radius~$(\delta_\alpha - 1) \cdot x $ at~$p_1$ and 
pay~$ x ^\alpha + (\delta_\alpha - 1)^\alpha \cdot x ^\alpha $. 
Since the competitive ratio of~\alg is~$c$, 
we have 
that~$\delta_\alpha^\alpha \cdot x^\alpha 
		\leqslant c\cdot x ^\alpha( 1+ (\delta_\alpha - 1)^\alpha)  + a$ and hence 
\begin{align*}
c \geqslant \frac{
	\delta_\alpha^\alpha 
	}{
	1 + (\delta_\alpha - 1)^\alpha 
	}
	-\frac{
	a
	}{
	x ^\alpha ( 1 + (\delta_\alpha - 1)^\alpha)
	}.
\end{align*}
Since the second term can be made arbitrarily 
small by choosing~$x$ large enough,~$c$ must be 
at least~$\frac{
	\delta_\alpha^\alpha 
	}{
	(\delta_\alpha - 1)^\alpha 
	}$.

In the latter case,~\alg 
has one disk of radius at least~$x$ and one of radius at least~$(\delta_\alpha - 1) \cdot x $ 
before~$p_3$ is inserted. 
We split this case into two subcases: in the first one, \alg increases the 
radius of the disk at~$p_0$ and in the second one, \alg increases the radius of the disk at~$p_1$. 
The cost~\alg has to pay after~$p_3$ has been inserted is at 
least either~$\delta_\alpha^\alpha \cdot x^\alpha + (\delta_\alpha -1 )^\alpha \cdot x^\alpha$ in the first 
subcase, 
or~$x^\alpha + (\delta_\alpha +1 )^\alpha \cdot x^\alpha$ in the second, 
whereas the optimal solution for both subcases would 
be to place only one disk of radius~$\delta_\alpha \cdot x$ at~$p_0$ and pay~$\delta_\alpha^\alpha \cdot x^\alpha$. 
Since the competitive ratio of~\alg is~$c$, 
we have 
that~$\delta_\alpha^\alpha \cdot x^\alpha + (\delta_\alpha -1 )^\alpha \cdot x^\alpha 
		\leqslant c \cdot \delta_\alpha^\alpha \cdot x^\alpha + a$ for the first subcase and hence 
\begin{align*}
c \geqslant \frac{
	\delta_\alpha^\alpha + (\delta_\alpha -1 )^\alpha 
	}{
	\delta_\alpha^\alpha 
	}
	-\frac{
	a
	}{
	\delta_\alpha^\alpha \cdot x^\alpha
	};
\end{align*}
and that~$x^\alpha + (\delta_\alpha +1 )^\alpha \cdot x^\alpha
		\leqslant c\delta_\alpha^\alpha \cdot x^\alpha + a$ for the second subcase and hence 
\begin{align*}
c \geqslant \frac{
	1 + (\delta_\alpha +1 )^\alpha 
	}{
	\delta_\alpha^\alpha 
	}
	-\frac{
	a
	}{
	\delta_\alpha^\alpha \cdot x^\alpha
	}.
\end{align*}
Since, in both subcases, the second term can be made arbitrarily 
small by choosing~$x$ large enough,~$c$ must be 
at least~$\frac{
	\delta_\alpha^\alpha + (\delta_\alpha -1 )^\alpha 
	}{
	\delta_\alpha^\alpha 
	}$ for the first subcase, 
and at least~$\frac{
	1 + (\delta_\alpha +1 )^\alpha 
	}{
	\delta_\alpha^\alpha 
	}$, otherwise there is an infinite family 
of instances contradicting the competitive ratio 
for these two subcases. 

Therefore, the competitive ratio of~\alg must 
be at least the minimum of the competitive ratio 
between these cases, which is exactly~$c_\alpha$. 
Even though it is not clear how to compute the value of~$c_\alpha$ for any fixed~$\alpha>1$, 
it is easy to see it is strictly bigger than~$1$. 
If~$\alpha = 2$, we have 
\begin{align*}
c_2 :&= \max_{\delta>1} \min \left(
		\frac{\delta^2}{1+(\delta-1)^2}, 
		\frac{\delta^2 + (\delta-1)^2}{\delta^2}, 
		\frac{1 + (\delta+1)^2}{\delta^2}
		\right) \\
	&= \frac{1}{12}\left( 
		4 + \sqrt[3]{496-24\sqrt{183}} + 2 \sqrt[3]{62 + 3 \sqrt{183}}
		\right) \\
	&\approx 1.58 
\end{align*}
which is achieved for 
\begin{align*}
\delta_2 :&= \argmax_{\delta>1}  \min \left(
		\frac{\delta^2}{1+(\delta-1)^2}, 
		\frac{\delta^2 + (\delta-1)^2}{\delta^2}, 
		\frac{1 + (\delta+1)^2}{\delta^2}
		\right) \\
	&= \frac{1}{3} \left(
		5 + \sqrt[3]{62-3\sqrt{183}} +  \sqrt[3]{62 + 3 \sqrt{183}}
		\right) \\
	&\approx 4.15.
\end{align*} 
\end{proof}

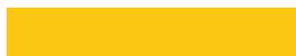
\mypara{Bounds for \NN and \CI{}.}
We now prove bounds on the competitive ratio of the algorithms \NN  and \CI
explained in the introduction. 
\begin{theorem}
Consider the range-assignment problem in~$\Reals^1$ with distance-power
gradient~$\alpha$.
\vspace*{-12mm}
\begin{quotation} \noindent
\myenumerate{
\item[{\rm\sf (i)}]  For any~$\alpha>1$, the competitive ratio of \CI is at most~$2$. 
\item[{\rm\sf (ii)}]  For any~$\alpha>1$, the competitive ratio of \NN is exactly~$2$.
}
\end{quotation}
\end{theorem}
\begin{proof}
We first prove the upper bounds. 
Assume without loss of generality that~$p_0=0$.
We first prove that both NN and \CI perform optimally for~$\alpha>1$ on any 
sequence~$p_0,p_1,\ldots,p_{n-1}$ with $p_j\geq 0$ for all $1\leq j<n$.
\begin{claiminproof}
Suppose $p_0=0$ and $p_j\geq 0$ for all $1\leq j<n$. Then 
\NN and \CI are optimal.
\end{claiminproof}
\begin{proofinproof}
We first observe that  for any point $p_j$ the following holds for the
graph~$\graph_{r_j}$ that we have after the insertion of $p_j$: for any
point $p_i$ with $0<i\leq j$ there is a path from the source~$p_0$ to~$p_i$ 
that only uses edges directed from left to right, that is, edges~$(p_{i'},p_{i''})$
with $p_{i'}< p_{i''}$. Indeed, if the path uses an edge~$(p_{i'},p_{i''})$
with $p_{i'}> p_{i''}$ then the subpath from $p_0$ to $p_{i'}$ must contain an
edge~$(p_s,p_t)$ with $p_s\leq p_{i''}\leq p_t$, and then we can go directly
from $p_s$ to $p_{i''}$. This observation implies that there exists an optimal 
strategy \optalg such that the balls $B(p_i,r_j(p_i))$ of the currently inserted 
points never extend beyond the currently rightmost point, a property which holds for \NN and \CI as well.
(Intuitively, the part of $B(p_i,r_j(p_i))$ to the right of the rightmost point is currently
useless, and the part of $B(p_i,r_j(p_i))$ to the left of $p_i$ is not
needed because we never need edges going to the left. Hence, we decrease~$r_j(p_i)$
until the right endpoint of $B(p_i,r_j(p_i))$ coincides with the currently rightmost
point, and increase the range of $p_i$ later, as needed.) 

Now imagine running \NN, \CI, and \optalg simultaneously on~$\objects$.
We claim that \NN and \CI do exactly the same, and that their cost increase
after the insertion of any point~$p_j$ is at most the cost increase of~\optalg.
To see this, let $p_{j'}$ be the rightmost point just before inserting~$p_j$.
If $p_j < p_{j'}$ then \NN and \CI do not increase any range---since $p_{j'}$
is reachable from $p_0$, the point $p_j$ must already be reachable as well---and so
the cost increase is zero. If $p_j>p_{j'}$ then \NN and \CI both increase
the range of $p_{j'}$ from 0 to $p_j-p_{j'}$. For \NN this is clear. 
For \CI it follows from the fact that $\alpha>1$. Indeed, increasing the range
of some $p_i<p_{j'}$ gives a cost increase 
$(r_{j-1}(p_i)+x+(p_j-p_{j'}))^\alpha - (r_{j-1}(p_i))^\alpha$,
for some $x\geq 0$. This is more than $(p_j-p_{j'})^\alpha$, since we must
have $r_{j-1}(p_i) +x>0$. By a similar reasoning, and using that the
balls of \optalg do not extend beyond $p_{j'}$, we conclude that the cost increase
of \optalg cannot be smaller than $(p_j-p_{j'})^\alpha$.
Hence, \NN and \CI are optimal on a sequence of non-negative points.
\end{proofinproof}
Next, we prove that the optimality for non-negative points gives a competitive ratio of 
at most~$2$ for any input sequence~$\objects$. Let~$\objects^+$ and $\objects^-$ denote
the subsequences of~$\objects$ consisting of the points with non-negative 
and non-positive points, respectively.
Note that the source point~$p_0=0$ is included in both subsequences. We claim 
that~$\cost_{\alpha}(\optalg((\objects)) \geq \cost_{\alpha}(\optalg((\objects^+))$. 
Indeed, we can modify the optimal solution for~$\objects$ to a valid 
solution for $\objects^+$ whose cost is at most~$\cost_{\alpha}(\optalg((\objects))$, 
as follows: whenever the range of a point~$p_i \not\in \objects^+$ is increased 
to reach a point $p_j\in \objects^+$, we instead increase the range of $p_0$ by the 
same amount. A similar argument gives $\cost_{\alpha}(\optalg((\objects)) \geq \cost_{\alpha}(\optalg((\objects^-))$. 

We now argue that~$\cost_\alpha(\NN(\objects)) \leq \cost_{\alpha}(\NN(\objects^+))+ \cost_\alpha(\NN(\objects^-))$
and, similarly, that $\cost_\alpha(\CI(\objects)) \leq \cost_{\alpha}(\CI(\objects^+))+ \cost_\alpha(\CI(\objects^-))$.

Imagine running $\NN$ simultaneously on $\objects$, on $\objects^+$ and on $\objects^-$.
We claim that the increase of $\cost_\alpha(\NN(\objects))$ upon the arrival of a new point~$p_j$
is at most the increase of $\cost_{\alpha}(\NN(\objects^+))$ if $p_j>0$, and 
at most the increase of $\cost_{\alpha}(\NN(\objects^-))$ if $p_j<0$.
To see this, assume without loss of generality that $p_j>0$ and suppose 
the increase of $\cost_\alpha(\NN(\objects))$ is non-zero. Then $p_j$
lies to the right of the currently rightmost point,~$p_{i}$. Both $\NN(\objects)$
and $\NN(\objects^+)$ then increase the range of~$p_i$, and pay the same cost.
The only exception is when $i=0$, that is, $p_j$ is the first point with $p_j>0$.
In this case $\NN(\objects)$ may pay less than $\NN(\objects^+)$, since
$\NN(\objects)$ could already have increased the range of $p_0$ due to arrivals of points
to the left of~$p_0$. 

A similar argument works for \CI. Indeed,
$\CI(\objects^+)$ and $\CI(\objects^-)$ never extend a ball beyond the currently rightmost
and leftmost point, respectively. Hence, when the new point $p_j$ lies, say,
to the right of the currently rightmost point~$p_i$, then $\CI(\objects^+)$
would pay $(\dist(p_i,p_j))^\alpha$. Since $\CI(\objects)$ also has the option
to increase the range of $p_i$, it will never pay more.

Hence, for \NN---a similar computation holds for \CI---we get
\[
\cost_\alpha(\NN(\objects)) 
    \leq \cost_{\alpha}(\NN(\objects^+))+ \cost_\alpha(\NN(\objects^-))
    \leq 2\cdot \opt(\objects).
\]
It remains to prove the lower bound for part~(ii) of the theorem. Assume for a contradiction that
there is a constant~$a$ such that for all inputs~$\objects$ we have
$\cost_{\alpha}(\NN(P)) \leq (2-\eps) \cdot \cost_{\alpha}(\optalg(P))+a$. 
Consider the input $p_0=0$, $p_1=\delta x$, $p_2=x$, and $p_3=-x$,  
for some~$\delta \in (0,1]$ and $x>0$ to be determined later. 
The optimal solution has $r_3(p_0)=x$ and $r_3(p_1)=r_3(p_2)=r_3(p_3)=0$, 
while \NN has $r_3(p_0)=x$ and $r_3(p_1)=(1-\delta)x$ and $r_3(p_2)=r_3(p_3)=0$. 
Hence, the competitive ratio that \NN achieves on this instance is 
\[
c = 
\frac{((1-\delta)^\alpha +1) x^\alpha - a}{x^\alpha} = 
	(1-\delta)^\alpha +1 
	- \frac{a}{x^\alpha}, 
\]
which is larger than~$2-\varepsilon$ when we pick~$\delta$ sufficiently small 
and~$x$ sufficiently large. 
\end{proof}
 \section{Online range-assignment in $\Reals^2$}
\subsection{Bounds on the competitive ratio of \NN and \CI when~$\alpha >2$}
As before, let~$p_0,\ldots,p_{n-1}$ be the sequence of inserted points, 
with~$p_0$ being the source point.
Consider a fixed point~$p_i$, and a disk $D$ centered at $p_i$---the disk~$D$ need not 
have radius equal to the range of~$p_i$. Define
$\set(p_i,D) := \{ p_j : j\geq i \mbox{ and } p_j\in D\}$ to be the set
containing $p_i$ plus all points arriving after $p_i$ that lie in $D$. For a 
point $p_j$, define $\cost_{\alpha}(\NN{},p_j)$ to be the cost 
incurred by \NN when $p_j$ is inserted; in other words,~$\cost_{\alpha}(\NN,p_j) := 0$
when $p_j$ falls into an existing disk $B(p_i,r_{j-1}(p_i))$, and $\cost_{\alpha}(\NN,p_j) :=(r_{j}(p_k))^\alpha - (r_{j-1}(p_k))^\alpha$
otherwise, where~$p_k := \nn(p_j)$.
Define $\cost_{\alpha}(\CI,p_j)$ 
similarly for \CI{}. Finally, for $p_j\in \set(p_i,D)$ define 
\[
F_{\alpha} (p_j) = F_{\alpha} (p_j; p_i, D) 
		:= \min \{ \dist(p_j,p_k)^{\alpha} 
		\mid p_k \in \set(p_i,D) \mbox{ and } k < j \}.
\]
The next lemma shows that we can use the function $F_{\alpha}$ to 
upper bound the cost of \NN and \CI{}. We later apply this lemma 
to all disks in an optimal solution to bound the competitive ratio. 
Note that $\cost_\alpha(\NN,p_j) \leq F_\alpha(p_j)$. Indeed, \NN either pays 
zero (when $p_j$ already lies inside a disk) or it expands the disk of $p_j$'s 
nearest neighbor (which may or may not lie in $D$) which costs at 
most $F_\alpha(p_j)$. Similarly $\cost_\alpha(\CI,p_j) \leq F_\alpha(p_j)$. 
Hence we have:
\begin{lemma}\label{lem:NN-CI<F_alpha}
Let $p_i$ be any input point and $D$ a disk centered at~$p_i$. Then for any subset $S(D)\subseteq \set(p_i,D)\setminus \{p_i\}$ we have:
\[
\sum_{p_j\in \set(D) } \cost_{\alpha}(\NN ,p_j) \leq \sum_{p_j\in \set(D)  } F_{\alpha}(p_j) \qquad
\mbox{and}\qquad 
\sum_{p_j\in \set(D) } \cost_{\alpha}(\CI ,p_j) \leq \sum_{p_j\in \set(D)  } F_{\alpha}(p_j).
\]
\end{lemma}
Lemma~\ref{lem:NN-CI<F_alpha} suggests the following strategy to bound the 
competitive ratio of \NN (and \CI). Consider, for each point~$p_i$, the final 
disk~$D$ placed at~$p_i$ in an optimal solution, and let~$\rho$ be its radius. 
The cost of this disk is~$\rho^\alpha$. We charge the cost of the disks placed 
by \NN (or \CI) at points~$p_j$ inside~$D$---this cost can be bounded using the 
function~$F_{\alpha}$, by Lemma~\ref{lem:NN-CI<F_alpha}---to the 
cost of~$D$. This motivates the following definition:
\begin{align*}
F^*_\alpha := \max \frac{1}
	{\rho^\alpha}\sum_{p_j\in \set(D) } F_{\alpha}(p_j), 
\end{align*}
where the maximum is over any possible input instance $\objects$,
any point~$p_i\in\objects$, any disk~$D$ of radius~$\rho$ centered at $p_i$, and any subset $S(D)\subseteq \set(p_i,D)\setminus \{p_i\}$. The value $F^*_\alpha$
bounds the maximum total charge to any disk $D$
in the optimal solution, relative to $D$'s cost $\rho^\alpha$.
The next lemma shows that for~$\alpha>2$, the value~$F^*_{\alpha}$ 
is bounded by a constant (depending on~$\alpha$).
\begin{lemma}\label{lem:ub-F_alpha}
We have that $F^*_\alpha \leqslant \alpha \frac{2^\alpha -3}{2^{\alpha-1}-\alpha}$ 
for any~$\alpha > 2$.
\end{lemma}
The formal proof of the lemma is  quite technical 
\ifFullVersion
so we sketch the intuition here before diving into the proof, 
\fi
also showing why the condition $\alpha>2$ is needed.
The quantity $F^*_\alpha$ can be thought of in the following way. Consider a disk $D$
of radius $\rho$ centered at $p_i$, and imagine the points in $\set(D)$ arriving one by one.
(The points in $\set(p_i,D)\setminus \set(D)$ are irrelevant.) Whenever a new point~$p_j$
arrives, then $F^*_{\alpha}$ increases by $\dist(p_j,\nn(p_j))^\alpha$, where $\nn(p_j)$ is $p_j$'s
nearest neighbor among the already arrived points from $\set(D)$ including~$p_i$.
Since the more points arrive the distances to the nearest neighbor will decrease---more precisely,
we cannot have many points whose nearest neighbor is at a relatively large distance---the
hope is that the sum of these distance to the power $\alpha$ converges, and this is
indeed what we can prove for $\alpha>2$. For $\alpha=2$ it does not converge, as shown
by the following example.

Let $D$ be a unit disk centered at $p_i$, and consider the inscribed square~$\sigma$ of $D$.
Note that the radius of $\sigma$---the distance from its center to its vertices---is~1.
We insert a set $\set(D)$ of $n-1$ points in rounds, as follows.
In the first round we partition $\sigma$ into four squares of radius~1/2, and we insert
a point in the center of each of them. These four points all have $p_i$
as nearest neighbor, and $F^*_{\alpha}$ increases by $4\cdot(1/2)^\alpha=(1/2)^{\alpha-2}$. 
We recurse in each of the four squares.
Thus in the $k$-th round, we have~$4^{k-1}$ squares of radius~$(1/2)^{k-1}$, 
each of which is partitioned into four squares of radius~$(1/2)^k$, and we place a point inside each such subsquare.
This increases $F^*_\alpha$ by~$4^k\cdot (1/2^{k})^\alpha=(1/2^{2-\alpha})^k$. 
The total cost is $\sum_{k=1}^t (1/2^{2-\alpha})^k$, where $t:=\Theta(\log n)$ 
is the number of rounds. 

Note that  $1/2^{2-\alpha}=1$ for $\alpha=2$, giving $F^*_2=\Omega(\log n)$, while
for $\alpha>2$ the total cost converges.
Also note that the example only gives a lower bound on $F^*_2$, it does not show
that \NN has unbounded competitive ratio for~$\alpha=2$. The reason is that \NN actually
pays less than $F^*_2$, since most points~$p_j$ are already within range of an existing point
upon insertion, and so we do not have to pay~$\dist(p_j,\nn(p_j))^\alpha$.
Indeed, in the next section we prove, using a different argument, that \NN is
$O(1)$-competitive even for $\alpha=2$.

\ifFullVersion
We now present the proof of Lemma~\ref{lem:ub-F_alpha}.
\begin{proofof}{Lemma~\ref{lem:ub-F_alpha}}
Let~$p_\ell$ be a point and let~$D$ be any disk  
centred at $p_\ell$. 
For the sake of simplicity, we rescale~$D$ to be a unit disk 
and relabel points in~$\set(D)$ as~$p_0,\ldots,p_k$ without changing 
the ordering and where~$p_0$ is the center of~$D$. 
We show that~$ \sum_{i=1}^{k}F_{\alpha} (p_i) \leqslant \alpha \frac{2^\alpha -3}{2^{\alpha-1}-\alpha}$. 
To that purpose we create a potential 
function~$\Phi:\{0,\ldots,k \} \to \Reals $, with~$\Phi(i)$ being 
the potential when~$p_i$ is inserted, 
with the following properties: 
\begin{itemize}
\item $\Phi(0)=\alpha \frac{2^\alpha -3}{2^{\alpha-1}-\alpha}$, 
\item $\Phi(i)>0$ for any~$i=0,\ldots,k$, 
\item $\Phi(i-1)-\Phi(i) \geqslant F_\alpha(p_i)$ for any~$i=1,\ldots,k$. 
\end{itemize} 
If such a potential function exists, we then indeed 
have~$F^*_{\alpha} (p_i) \leqslant \alpha \frac{2^\alpha -3}{2^{\alpha-1}-\alpha}$. 

Let~$p_i$ be the last point inserted. 
For any point~$q$ in the plane, let~$\nn_i(q)$ be its closest point 
with among~$p_1,\ldots,p_i$. We define 
the potential~$\phi(q,i)$ at~$q$ at time~$i$ as follows: 
\begin{align*}
\phi(q,i) := \left\lbrace \begin{array}{l}
c_\alpha
		\dist(q,\nn_i(q))^{\alpha-2} \\[4pt]
		\qquad \mbox{if } q \in D \mbox{, that is, } \dist(q,p_0) \leq 1; \\[12pt]
c_\alpha
		( \dist(q,\nn_i(q))^{\alpha-2} 
			- \dist(q,\partial D )^{\alpha -2} )\\[4pt]
		\qquad \mbox{if } 1 < \dist(q,p_0) \leqslant \frac{3}{2}, 
		\mbox{ and } \dist(q,\nn_i(q)) \geqslant \dist(q,\partial D_2); \\[12pt]
c_\alpha
		( d(q,\partial D_2)^{\alpha-2}
			 - \dist(q,\partial D )^{\alpha -2} ) \\[4pt]
		\qquad \mbox{if } 1 < \dist(q,p_0) \leqslant \frac{3}{2},
		 \mbox{ and } \dist(q,\nn_i(q)) < \dist(q,\partial D_2);  \\[12pt]
0  \quad \ \ \mbox{otherwise;}
\end{array} \right.
\end{align*}
where~$D_2$ is the disk of radius 2 centred at~$p_0$ 
and~$c_\alpha = \frac{\alpha(\alpha -1) 2^{\alpha-2}}{\pi (2^{\alpha-1} -\alpha)}$ 
is a constant depending only on~$\alpha$. 
See Figure~\ref{fig:phi} for an illustration of the cases. 
\begin{figure}
\begin{center}
\begin{tikzpicture}
\node at (0,0) (p0) {};
\node at (-0.8,0.4) (pi) {};
\node at (0.2,-0.7) (pj) {};
\node at (0.5,0.8) (nn) {};
\node at (1.6,1.3) (q) {};

\draw[opacity=0, fill opacity=1, fill=mygray-light] (p0) circle (2.25);
\draw (p0) circle (1.5);
\draw (p0) circle (3);

\node at (-1.5,1) {$D$};
\node at (-3.25,1.5) {$D_2$}; 

\draw (q.center) -- (nn.center);
\begin{scope}
    \clip(p0) circle (3);
	\draw (q.center) -- (3.2,2.6);
\end{scope}

\foreach \i/\l in {p0/p_0,nn/\nn_i(q),q/q,pi/p_i,pj/p_j}{
	\draw[fill=white] (\i) circle (0.06);
	\node at ([yshift=-4mm]\i.center) {$\l$};
}
\end{tikzpicture}
\end{center}
\caption{Outside the grey region the function~$\phi$ is always 0. 
When~$q$ is inside~$D$, the function~$\phi(q,i)$ is 
simply~$c_\alpha \dist(q,\nn_i(q))^{\alpha-2}$. Finally, when~$q$ is in the grey area 
but not in~$D$, that is~$1< \dist(q,p_0) \leqslant 1.5$, we have 
that~$\phi(q,i)= c_\alpha \left(\min \{ \dist(q,\nn_i(q)),\dist(q,\partial D_2) \}^{\alpha-2} - \dist(q,\partial D )^{\alpha -2} \right)$ }
\label{fig:phi}
\end{figure}
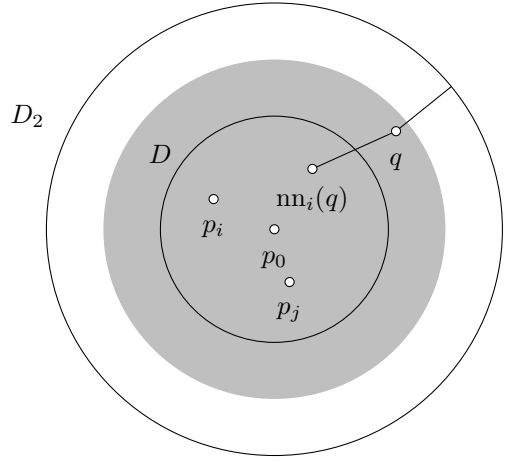

We finally define the potential function at time~$i$ as 
$$\Phi(i):= \iint_{\Reals^2} \phi (q,i) dq.$$
This potential function can be interpreted as a volume in~$\Reals^3$, 
where we assume without loss of generality that the center of~$D$ lies at 
the origin and the points~$p_0,\ldots,p_k$ all lie in the plane~$z=0$. 
The volume then 
consists of the following. Over~$D$ it is the volume between the plane~$z=0$ 
and the lower envelope of a set of ``paraboloids'', one for each 
point~$p_j\in \{ p_0,\ldots, p_i \}$, defined 
by~$\parb_\alpha(p_j)= \{ (x,y,z) \mid z = c_\alpha \dist((x,y),p_j)^{\alpha-2} \}$. 
Outside of~$D$, on the other hand, the 
volume is defined as the volume under the lower envelope of~$\parb_\alpha(p)$ 
for each point~$p\in \{ p_0,\ldots, p_i\} \cup \partial D_2$ 
and above the paraboloids~$\parb_\alpha(p)$ for each point~$p\in \partial D$. 
See Figure~\ref{fig:V_i} for an illustration.
\begin{figure}
\begin{center}
\begin{tikzpicture}
\draw[thick] (-5,0) -- (5,0);
\draw[thick, -latex] (0,0) -- (0,2.5);
\node at (0,2.8) {$z$};

\node at (0.3,-0.3) {$p_0$};
\node at (1.6,-0.3) {$p_1$};
\node at (-0.6,-0.3) {$p_2$};

\draw[thick, dotted] (0,0) --++ (0,-0.85);
\draw[thick, dotted] (2,0) --++ (0,-0.85);
\draw[thick, dotted] (3,0) --++ (0,-0.85);
\draw[thick, dotted] (4,0) --++ (0,-0.85);
\draw[thick, dotted] (-2,0) --++ (0,-0.85);
\draw[thick, dotted] (-3,0) --++ (0,-0.85);
\draw[thick, dotted] (-4,0) --++ (0,-0.85);

\draw[latex-latex] (0,-0.7) -- (2,-0.7);
\draw[latex-latex] (2,-0.7) -- (3,-0.7);
\draw[latex-latex] (3,-0.7) -- (4,-0.7);
\draw[latex-latex] (0,-0.7) -- (-2,-0.7);
\draw[latex-latex] (-2,-0.7) -- (-3,-0.7);
\draw[latex-latex] (-3,-0.7) -- (-4,-0.7);

\node at (1,-1) {$1$}; 
\node at (-1,-1) {$1$}; 
\node at (2.5,-1) {$\frac{1}{2}$}; 
\node at (3.5,-1) {$\frac{1}{2}$}; 
\node at (-2.5,-1) {$\frac{1}{2}$}; 
\node at (-3.5,-1) {$\frac{1}{2}$}; 

\begin{scope}
	\clip(0,0) parabola (2,2) -- (4,2) -- (4,-0.1) -- 
			(-4, -0.1) -- (-4,0) parabola (-2,2) -- cycle;
	\clip(1.6,0) parabola (3.6,2) -- (4,2) -- (4,-0.1) -- 
			(-4, -0.1) -- (-4,0) -- (-4,2) -- (1.6,2) -- cycle;
	\clip(-0.6,0) parabola (1.4,2) -- (4,2) -- (4,-0.1) -- 
			(-4, -0.1) -- (-4,0) -- (-4,2) -- (1.6,2) -- cycle;
	\clip(0,0) parabola (-2,2) -- (-4,2) -- (-4,-0.1) -- 
			(4, -0.1) -- (4,0) parabola (2,2) -- cycle;
	\clip(1.6,0) parabola (-0.4,2) -- (-4,2) -- (-4,-0.1) -- 
			(4, -0.1) -- (4,0) -- (4,2) -- (4,1.6) -- cycle;
	\clip(-0.6,0) parabola (-2.6,2) -- (-4,2) -- (-4,-0.1) -- 
			(4, -0.1) -- (4,0) -- (4,2) -- (4,1.6) -- cycle;
	\clip(2,0) parabola (4,2) -- (-4,2) -- (-4,-0) -- cycle;
	\clip(-2,0) parabola (-4,2) -- (4,2) -- (4,-0) -- cycle;
    \draw[fill=mygray-dark, opacity=0, fill opacity=1] 
   		(4,2) rectangle (-4,0); 
\end{scope}

\draw (0,0) parabola (2,2);
\draw (0,0) parabola (-2,2);
\draw (1.6,0) parabola (3.6,2);
\draw (1.6,0) parabola (-0.4,2);
\draw (-0.6,0) parabola (1.4,2);
\draw (-0.6,0) parabola (-2.6,2);
\draw (2,0) parabola (4,2);
\draw (-2,0) parabola (-4,2);
\draw (4,0) parabola (2,2);
\draw (-4,0) parabola (-2,2);

\foreach \i in {0,1.6,-0.6}{
	\draw[fill=white, thick] (\i,0) circle (0.06);
}
\end{tikzpicture}
\end{center}
\caption{Cross section of the volume~$V_i$ in gray. For clarity, we do not draw 
the paraboloids of points outside the cross section. }
\label{fig:V_i}
\end{figure}
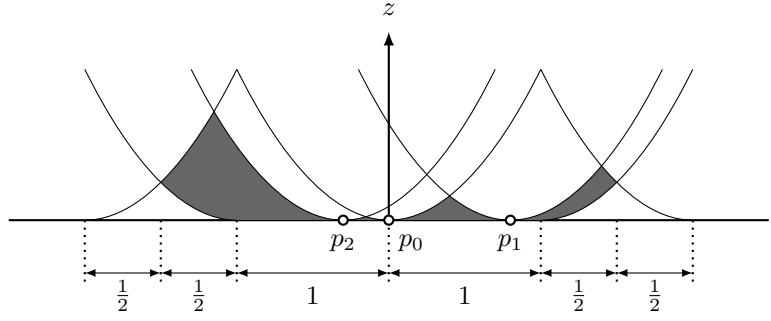

We now need to show that this potential function has the claimed properties. 
It is easy to see that~$\Phi(i)>0$ for each~$i=0,\ldots,k$. Next we show 
that the decrease of potential is at least as large as the cost. 
Let~$p_i$ be the point inserted and let~$p_j$, with~$j<i$, be 
a nearest neighbor of~$p_i$ with~$d^*:=\dist(p_i,p_j)$. Upon insertion of~$p_i$, we add 
a paraboloid defined at a point~$q\in D$ by~$c_\alpha \dist(q,p_i)^{\alpha -2}$. 
The decrease of potential is then the volume~$V_i$ subtracted by this 
surface. Let us consider the 
volume 
\begin{align*}
V_i^*:=\{ q=(x,y,z)\mid &\dist((x,y,0),p_i) \leqslant d^* \\
		&\mbox{ and } c_\alpha \dist(q,p_i)^{\alpha -2} \leqslant z \leqslant 
		c_\alpha (d^*-\dist(q,p_i))^{\alpha -2}  \}. 
\end{align*}
See Figure~\ref{fig:V_i_star} 
for an illustration of the volume~$V_i^*$. 
\begin{figure}
\begin{center}
\begin{tikzpicture}
\draw[thick] (-3,0) -- (3,0);
\node at (0,2.8) {$z$};
\node[scale=0.75] at (0,2.3) {\AxisRotator[rotate=-90]};

\begin{scope}
	\clip(0,0) parabola (2,2) -- (-2,2) -- (-2,0) -- cycle;
	\clip(0,0) parabola (-2,2) -- (2,2) -- (2,0) -- cycle;
	\clip(2,0) parabola (0,2) -- (-2,2) -- (-2,0) -- cycle;
	\clip(-2,0) parabola (0,2) -- (2,2) -- (2,0) -- cycle;
    \draw[fill=mygray-dark, opacity=0, fill opacity=1] 
   		(2,2) rectangle (-2,0); 
\end{scope}

\draw[thick, -latex] (0,0) -- (0,2.5);

\node at (0,-0.3) {$p_i$};
\draw[thick] (1,0.1) -- (1,-0.1);
\draw[thick] (2,0.1) -- (2,-0.1);
\draw[thick] (-1,0.1) -- (-1,-0.1);
\draw[thick] (-2,0.1) -- (-2,-0.1);
\node at (0.5,-0.5) {$\frac{d^*}{2}$}; 
\node at (1.5,-0.5) {$\frac{d^*}{2}$}; 
\node at (-0.5,-0.5) {$\frac{d^*}{2}$}; 
\node at (-1.5,-0.5) {$\frac{d^*}{2}$}; 

\draw (0,0) parabola (2,2);
\draw (0,0) parabola (-2,2);
\draw (2,0) parabola (0,2);
\draw (-2,0) parabola (0,2);

\foreach \i in {0}{
	\draw[fill=white, thick] (\i,0) circle (0.06);
}
\end{tikzpicture}
\begin{tikzpicture}
\begin{scope}
	\clip(0,0) parabola (2,2) -- (-2,2) -- (-2,0) -- cycle;
	\clip(0,0) parabola (-2,2) -- (2,2) -- (2,0) -- cycle;
	\clip(2,0) parabola (0,2) -- (-2,2) -- (-2,0) -- cycle;
	\clip(-2,0) parabola (0,2) -- (2,2) -- (2,0) -- cycle;
    \draw[fill=black, opacity=0, fill opacity=0.2] 
   		(2,2) rectangle (-2,0); 
   	\shade[left color=black!10,right color=black!80] (2,2) rectangle (0,0);
   	\shade[left color=black!80,right color=black!10] (0,2) rectangle (-2,0);
\end{scope}

\begin{scope}[even odd rule]
	\clip(0,0) parabola (2,2) -- (-2,2) -- (-2,0) -- cycle;
	\clip(0,0) parabola (-2,2) -- (2,2) -- (2,0) -- cycle;
	\clip(2,0) parabola (0,2) -- (-2,2) -- (-2,0) -- cycle;
	\clip(-2,0) parabola (0,2) -- (2,2) -- (2,0) -- cycle;
	\clip(0,0.5) ellipse (2cm and 1cm) (0,0.54) ellipse (1cm and 0.2cm);
	\clip(1.1, 0.5) rectangle (-1.1, -0.1);
    \shade[inner color=black!25,outer color=black!70] 
   		(1,1.3) rectangle (-1.2,-0.9); 
\end{scope}

\begin{scope}
	\clip(1.1, 0.5) rectangle (-1.1, 0);
	\draw[very thin] (0,0.54) ellipse (1cm and 0.2cm); 
\end{scope}
\begin{scope}
	\clip(1.1, 0.5) rectangle (-1.1, 1);
	\draw[very thin, color=black!50] (0,0.46) ellipse (1cm and 0.2cm); 
\end{scope}

\draw[thick, -latex] (-1.5,0) -- (1.5,0);
\node at (1.8,0) {$y$};

\draw[thick, -latex] (0,2) -- (0,2.5);
\draw[color=black!50] (0,0) -- (0,2);
\draw[thick] (0,-0.1) -- (0,0);
\node at (0,2.8) {$z$};

\draw[thick, -latex] (0,0) -- (-0.8,-0.5);
\draw[color=black!50] (0.4,0.25) -- (0,0);
\node at (-1,-0.625) {$x$};

\node at (0,-0.3) {$p_i$};

\node at (1.5,-0.5) {\textcolor{white}{$\frac{d^*}{2}$}}; 
\end{tikzpicture}
\end{center}
\caption{The volume we use as a lower bound on the decrease of potential 
upon insertion of~$p_i$. On the left, we have a cross section of the volume, 
where~$d^*=\dist(p_i,p_j)$ and~$p_j$ is the nearest point to~$p_i$ with~$j<i$. 
On the right, we have the volume in 3 dimensions. All the paraboloids 
defined by~$\parb_\alpha(p)$ for some~$p$.}
\label{fig:V_i_star}
\end{figure}
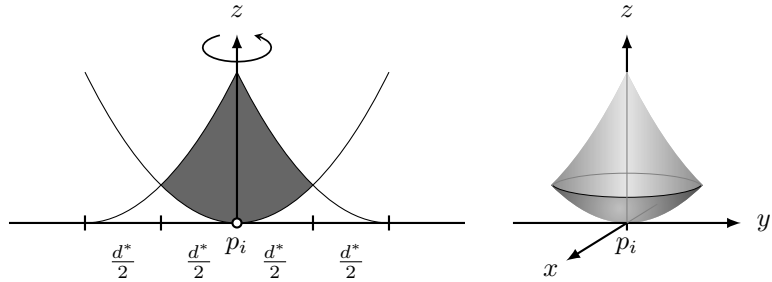
Next we argue that~$V_i^* \subseteq V_i$ by showing that the upper boundary of~$V_i^*$ is 
under the upper boundary of~$V_i$ and the lower boundary of~$V_i^*$ is 
above the lower boundary of~$V_i$. 
Since the upper boundary of~$V_i^*$ is defined by paraboloids at distance~$d^*$, 
and since~$d^*$ is the distance to the closest point~$p_j$, the volume~$V_i^*$ 
is under the upper boundary of~$V_i$. On the other hand, since the lower 
boundary of~$V_i^*$ is defined again by the same paraboloids, even 
if~$p_i$ is close to the boundary of~$D$, the volume~$V_i^*$ is above 
the lower boundary of~$V_i$. Therefore~$V_i^* \subseteq V_i$. 

We now compute the volume of~$V_i^*$. We do this by fixing a radius~$\rho$, 
then computing the area of the largest cylinder of radius~$\rho$ centered 
around the vertical axis passing through~$p_i$ and inscribed in~$V_i^*$ 
and integrating that value from~$0$ until~$d^*/2$. For a 
certain~$0 \leqslant \rho \leqslant d^*/2$, the area of the cylinder is~$2\pi \rho h(\rho)$ 
where~$h(\rho)$ is the height of the tallest cylinder 
of radius~$\rho$ inscribed in~$V_i^*$. 
It remains to compute~$h(\rho)$. This is given by the difference of height 
between the two paraboloids (one on~$p_i$ and one on a point at distance~$d^*$ 
of~$p_i$), i.e.,~$h(\rho)= c_\alpha ((d^*-\rho)^{\alpha-2} - \rho^{\alpha-2})$. 
Thus, 
\begin{align*}
\vol(V_i^*) &= \int_0^{d^*/2} 2\pi \rho c_\alpha ((d^*-\rho)^{\alpha-2} - \rho^{\alpha-2}) d\rho \\
&= 2\pi c_\alpha \int_0^{d^*/2} \rho(d^*-\rho)^{\alpha-2} - \rho^{\alpha-1} d\rho.
\end{align*}
We can integrate~$\rho(d^*-\rho)^{\alpha-2}$ by parts: 
\begin{align*}
\int_0^{d^*/2} \rho(d^*-\rho)^{\alpha-2} d\rho 
	&= \rho  \frac{-(d^*-\rho)^{\alpha-1}}{\alpha -1} \Big|_0^{d^*/2} 
		- \int_0^{d^*/2} 1\frac{-(d^*-\rho)^{\alpha-1}}{\alpha-1} d\rho \\
	&= - \rho  \frac{(d^*-\rho)^{\alpha-1}}{\alpha -1}  
		-  \frac{(d^*-\rho)^{\alpha}}
			{\alpha(\alpha-1)} \Big|_0^{d^*/2}.
\end{align*}
It gives us the following: 
\begin{align*}
\vol(V_i^*) &= 2\pi c_\alpha 
		\left[ - \rho \frac{(d^*-\rho)^{\alpha-1}}{\alpha -1}  
		-  \frac{(d^*-\rho)^{\alpha}}{\alpha(\alpha-1)}  
		- \frac{\rho^\alpha}{\alpha} \right]\Big|_0^{d^*/2} \\
	&= 2\pi c_\alpha 
		\left[ - \frac{d^*}{2} \frac{(d^*/2)^{\alpha-1}}{\alpha -1}  
		-  \frac{(d^*/2)^{\alpha}}{\alpha(\alpha-1)}  
		- \frac{(d^*/2)^\alpha}{\alpha} 
		+ \frac{d^{*\alpha}}{\alpha(\alpha -1)} \right] \\
	&= 2\pi c_\alpha 
		\left[ -  \frac{d^{*\alpha}}{2^{\alpha}(\alpha -1)}  
		-  \frac{d^{*\alpha}}{2^{\alpha}\alpha(\alpha-1)}  
		- \frac{d^{*\alpha}}{2^{\alpha}\alpha} 
		+ \frac{d^{*\alpha}}{\alpha(\alpha -1)} \right] \\
	&= \frac{\pi}{2^{\alpha-1}\alpha(\alpha-1)} c_\alpha d^{*\alpha}  
		\left[ - \alpha -1 -(\alpha-1) +2^\alpha \right] \\
	&= \frac{\pi(2^{\alpha-1} - \alpha)}{2^{\alpha-2}\alpha(\alpha-1)} 
		c_\alpha d^{*\alpha}. 
\end{align*}
Recall that~$c_\alpha=\frac{\alpha(\alpha -1) 2^{\alpha-2}}{\pi (2^{\alpha-1} -\alpha)}$. 
We thus get 
\begin{align*}
\vol(V_i^*) &=  d^{*\alpha}, 
\end{align*}
which is exactly the cost of inserting~$p_i$, therefore the decrease of potential 
is indeed at most the cost. 

It remains to show that~$\Phi(0)=\alpha \frac{2^\alpha -3}{2^{\alpha-1}-\alpha}$. 
To that purpose, let~$V_0$ be a volume representing~$\Phi(0)$, defined as 
\begin{align*}
V_0 := &\{ (x,y,z) \mid (x,y) \in D \mbox{ and } 
			0\leqslant z \leqslant c_\alpha \dist((x,y),p_0)^{\alpha -2} \} \\ 
	\cup &\{ (x,y,z) \mid 1< \dist((x,y),p_0) \leqslant \frac{3}{2} \\
	&\mbox{ and } 
			c_\alpha (2-\dist((x,y),p_0))^{\alpha -2} \leqslant 
			z \leqslant c_\alpha (\dist((x,y),p_0)-1)^{\alpha-2} \}
\end{align*}
as depicted in Figure~\ref{fig:V_0}. 
\begin{figure}
\begin{center}
\begin{tikzpicture}
\draw[thick] (-5,0) -- (5,0);
\node at (0,2.8) {$z$};
\node[scale=0.75] at (0,2.3) {\AxisRotator[rotate=-90]};

\begin{scope}
	\clip(0,0) parabola (2,2) -- (4,2) -- (4,-0.1) -- 
			(-4, -0.1) -- (-4,0) parabola (-2,2) -- cycle;
	\clip(0,0) parabola (-2,2) -- (-4,2) -- (-4,-0.1) -- 
			(4, -0.1) -- (4,0) parabola (2,2) -- cycle;
	\clip(2,0) parabola (4,2) -- (-4,2) -- (-4,-0) -- cycle;
	\clip(-2,0) parabola (-4,2) -- (4,2) -- (4,-0) -- cycle;
    \draw[fill=mygray-dark, opacity=0, fill opacity=1] 
   		(4,2) rectangle (-4,0); 
\end{scope}

\draw[thick, -latex] (0,0) -- (0,2.5);

\node at (0,-0.3) {$p_0$};
\draw[thick] (2,0.1) -- (2,-0.1);
\draw[thick] (3,0.1) -- (3,-0.1);
\draw[thick] (4,0.1) -- (4,-0.1);
\draw[thick] (-2,0.1) -- (-2,-0.1);
\draw[thick] (-3,0.1) -- (-3,-0.1);
\draw[thick] (-4,0.1) -- (-4,-0.1);
\node at (1,-0.5) {$1$}; 
\node at (-1,-0.5) {$1$}; 
\node at (2.5,-0.5) {$\frac{1}{2}$}; 
\node at (3.5,-0.5) {$\frac{1}{2}$}; 
\node at (-2.5,-0.5) {$\frac{1}{2}$}; 
\node at (-3.5,-0.5) {$\frac{1}{2}$}; 

\draw (0,0) parabola (2,2);
\draw (0,0) parabola (-2,2);
\draw (2,0) parabola (4,2);
\draw (-2,0) parabola (-4,2);
\draw (4,0) parabola (2,2);
\draw (-4,0) parabola (-2,2);

\foreach \i in {0}{
	\draw[fill=white, thick] (\i,0) circle (0.06);
}
\end{tikzpicture}
\end{center}
\caption{The volume~$V_0$. }
\label{fig:V_0}
\end{figure}
We use the same technique as above to 
compute $$\vol(V_0)=\int_0^{3/2} 2\pi \rho \cdot h(\rho)d\rho. $$
We have~$h(\rho)= c_\alpha \rho^{\alpha-2}$ when~$\rho\leqslant 1$ 
and~$h(\rho)= c_\alpha [(2-\rho)^{\alpha-2} - (\rho-1)^{\alpha-2}]$ 
when~$1 < \rho \leqslant 3/2$. We therefore get 
\begin{align*}
\vol(V_0) &= 2\pi c_\alpha \left( 
			\int_0^{1} \rho^{\alpha-1} d\rho 
			+ \int_1^{3/2} \rho(2-\rho)^{\alpha-2} - \rho(\rho-1)^{\alpha-2} d\rho 
			\right).
\end{align*}
We again use integration by parts. 
\begin{align*}
\int_1^{3/2} \rho(\rho-1)^{\alpha-2} d\rho 
		&= \rho \frac{(\rho-1)^{\alpha-1}}{\alpha-1} \Big|_1^{3/2}
			- \int_1^{3/2} 1\frac{(\rho-1)^{\alpha-1}}{\alpha-1} d\rho \\
		&= \left( \rho \frac{(\rho-1)^{\alpha-1}}{\alpha-1} 
				- \frac{(\rho-1)^{\alpha}}{\alpha(\alpha-1)}
			\right)\Big|_1^{3/2} \\
		&= \frac{3}{2} \frac{(1/2)^{\alpha-1}}{\alpha-1} 
				- \frac{(1/2)^{\alpha}}{\alpha(\alpha-1)} \\
		&= \frac{3}{2^\alpha (\alpha-1)} 
				- \frac{1}{2^{\alpha}\alpha(\alpha-1)} \\
		&= \frac{3\alpha - 1}{2^\alpha \alpha (\alpha-1)} 
\end{align*}
and 
\begin{align*}
\int_1^{3/2} \rho(2-\rho)^{\alpha-2} d\rho 
		&= \rho \frac{-(2-\rho)^{\alpha-1}}{\alpha-1} \Big|_1^{3/2}
			+ \int_1^{3/2} 1\frac{(2-\rho)^{\alpha-1}}{\alpha-1} d\rho \\
		&= \left( -\rho \frac{(2-\rho)^{\alpha-1}}{\alpha-1} 
				- \frac{(2-\rho)^{\alpha}}{\alpha(\alpha-1)}
			\right)\Big|_1^{3/2} \\
		&= - \frac{3}{2} \frac{(1/2)^{\alpha-1}}{\alpha-1} 
				- \frac{(1/2)^{\alpha}}{\alpha(\alpha-1)} 
				+ \frac{1}{\alpha-1} 
				+ \frac{1}{\alpha(\alpha-1)} \\
		&= - \frac{3}{2^\alpha (\alpha-1)} 
				- \frac{1}{2^{\alpha}\alpha(\alpha-1)}
				+ \frac{1}{\alpha-1} + \frac{1}{\alpha(\alpha-1)} \\
		&= \frac{-3\alpha -1 + 2^\alpha \alpha + 2^ \alpha }
			{2^\alpha \alpha(\alpha-1)}.
\end{align*}
This together with~$\int_0^{1} \rho^{\alpha-1} d\rho = 1/\alpha$ 
gives us the following: 
\begin{align*}
\vol(V_0) &= 2\pi c_\alpha \left( 
			\frac{1}{\alpha}  
			+ \frac{-3\alpha -1 + 2^\alpha \alpha + 2^ \alpha }
				{2^\alpha \alpha(\alpha-1)}
			-  \frac{3\alpha - 1}{2^\alpha \alpha (\alpha-1)} 
			\right) \\
		&= \frac{\pi c_\alpha}{2^{\alpha-1} \alpha (\alpha-1)} 
			\left( 2^\alpha (\alpha-1) 
				-3\alpha - 1 +  2^\alpha (\alpha + 1) 
				- 3\alpha + 1 
			\right) \\
		&= \frac{\pi c_\alpha}{2^{\alpha-1} \alpha (\alpha-1)}
			\left( 2^{\alpha} (\alpha-1+\alpha+1)  - 6 \alpha  
			\right)\\
		&= \frac{\pi ( 2^{\alpha} \alpha - 3\alpha ) c_\alpha}
			{2^{\alpha-2} \alpha (\alpha-1)}. 			
\end{align*}
Again, with~$c_\alpha=\frac{\alpha(\alpha -1) 2^{\alpha-2}}{\pi (2^{\alpha-1} -\alpha)}$, 
we obtain 
\begin{align*}
\vol(V_0) &= \frac{\pi \alpha ( 2^{\alpha} - 3 )}
			{2^{\alpha-2} \alpha (\alpha-1)}
			\cdot
			\dfrac{\alpha(\alpha -1) 2^{\alpha-2}}{\pi (2^{\alpha-1} -\alpha)} \\
		&= 	\alpha \frac{ 2^{\alpha} - 3 }
			{ 2^{\alpha-1} -\alpha},
\end{align*}
concluding the proof.
\end{proofof}
\fi
\medskip

Using Lemma~\ref{lem:ub-F_alpha} we can prove a bound on the competitive ratio of \NN and \CI.
\begin{theorem}
Let $F^*_\alpha := \alpha \frac{2^\alpha -3}{2^{\alpha-1}-\alpha}$. For any~$\alpha>2$, 
the competitive ratio of \NN and \CI in $\Reals^2$ is at
most~$\min \{ F^*_\beta \mid 2<\beta\leqslant \alpha\}$.
Hence, for $\alpha\leq \alpha^*$, where $\alpha^* = \argmin  F^*_\alpha \approx 4.3$,
the competitive ratio is at most $\alpha \frac{2^\alpha -3}{2^{\alpha-1}-\alpha}$, 
and for $\alpha>\alpha^*$ it is at most~12.94.
\end{theorem}
\begin{proof}
Consider a sequence~$p_0,p_1,\ldots,p_{n-1}$ of points in the plane. 
Let~$D_j$ be the disk centered at~$p_j$ in an optimal solution, 
after the last point~$p_{n-1}$ has been handled, and let~$\rho_j$ be its radius. 
Thus the cost of the optimal solution is
$\opt := \sum_{j=0}^{n-1} \rho_j^{\alpha}$. 
To bound the cost of \NN on the same sequence, we charge the 
cost of inserting~$p_i$, with $0<i\leq n-1$, to a disk~$D_j$ 
such that $j<i$ and $p_i\in D_j$. Such a disk $D_j$ exists, since 
after $p_i$'s insertion, $p_i$ is contained in a disk of an existing point $p_j$ and so $p_i$ will also be contained in~$D_j$,
the final disk of~$p_j$.
(If there are more such points, we take an arbitrary one.) 
Let~$\set(D_j)$ be the set of points that charge disk~$D_j$. 
Note that~$\{p_1,\ldots,p_{n-1}\} = \bigcup_{j=0}^{n-2} \set(D_j)$. 
Hence, using Lemmas~\ref{lem:NN-CI<F_alpha} and~\ref{lem:ub-F_alpha}, 
for any~$2 < \beta \leq \alpha$, for \NN (and similarly for \CI) we get:
\[
\begin{array}{llll}
\cost_\alpha(\NN) 
	&    = & \sum_{i=0}^{n-2}\sum_{p_j\in \set(D_i)} \cost_{\alpha}(\NN,p_j) & \\[2mm]
	&    = & \sum_{i=0}^{n-2} \rho_i^\alpha  \sum_{p_j\in \set(D_i)} 	\frac{\cost_{\alpha}(\NN,p_j)}{\rho_i^\alpha} & \\[2mm]
	& \leq & \sum_{i=0}^{n-2} \rho_i^\alpha  \sum_{p_j\in \set(D_i)} 	\frac{\dist(p_j,\nn(p_j))^\alpha}{\rho_i^\alpha} & \\[2mm]
	& \leq & \sum_{i=0}^{n-2} \rho_i^\alpha  \sum_{p_j\in \set(D_i)}    \frac{\dist(p_j,\nn(p_j))^\beta}{\rho_i^\beta} 
	        	&  \mbox{because } \dist(p_j,\nn(p_j)) \leq\rho \\[2mm]
	&   =  & \sum_{i=0}^{n-2} \rho_i^\alpha \sum_{p_j\in \set(D_i)} \frac{F_{\beta}(p_j)}{\rho_i^\beta} 
		& \mbox{by Lemma~\ref{lem:NN-CI<F_alpha} } \\[2mm]
	& \leq & \sum_{i=0}^{n-2} \rho_i^\alpha  F_{\beta}^* & \\[2mm]
	& \leq & \beta \frac{2^\beta -3}{2^{\beta-1}-\beta} \sum_{i=0}^{n-2} \rho_i^\alpha  
		& \mbox{by Lemma~\ref{lem:ub-F_alpha}}\\[2mm]
	& = & \beta \frac{2^\beta -3}{2^{\beta-1}-\beta} \opt{}. & 
\end{array}
\]
\end{proof}
The next theorem gives a lower bound on the competitive ratio of~\NN. 
\begin{theorem} \label{thm:lower-bound-NN-R2}
For any~$\alpha >1$, \NN has a competitive ratio of at 
least~$6 (1 + (\frac{\sqrt{6}-\sqrt{2}}{2})^\alpha) \approx 6 (1+ 0.52^\alpha)$ 
in the plane. In particular, for~$\alpha=2$, we get a lower bound of~$7.6$ on the competitive ratio.
\end{theorem}
\begin{proof}
Let $p_0$ be the source placed at the origin. 
The following construction is depicted in 
Figure~\ref{fig:lower_bound_NN}. We  
place~$p_1,\ldots,p_{18}$ in a disk of radius~$1$ around~$p_0$ 
as explained next, such that a possible solution is to place that single 
disk and pay~$1$s. For simplicity, in the rest of the proof we use 
polar coordinates. Let $\varepsilon>0$ be a positive number. Let 
then~$p_1=(\varepsilon, 0)$, $p_2(\varepsilon, \pi/3)$, 
..., and~$p_6=(\varepsilon, 5\pi/3)$ be the next six points. 
\NN places a disk of radius~$\varepsilon$ on~$p_0$. 
Let further~$p_7=(1,0)$, $p_8=(1,\pi/3)$, ..., and~$p_{12}=(1,5\pi/3)$ 
be the next six points. Here \NN places six disks of radius~$1-\varepsilon$
centered around~$p_1,\ldots,p_6$, paying~$6(1-\varepsilon)^\alpha$. 
Finally, let~$p_{13}=(1,\pi/6-\varepsilon)$, 
$p_{14}=(1,3\pi/6-\varepsilon)$,..., 
and~$p_{18}=(1,11\pi/6-\varepsilon)$ be the last six points. 
\NN is now forced to place 6 disks of 
radius almost equal to the side of a 12-gon of radius 1, 
that is~$2\sin (\pi/12) -\delta$ 
for some~$\delta>0$ that tends to~$0$ as~$\varepsilon$ tends to~$0$.  

\begin{figure}
\begin{center}
\begin{tikzpicture}
\node at (0,0) (p0) {};

\draw[fill=mygray-light, opacity=0, fill opacity=1] (p0) circle (2);
\draw[fill=mygray-medium] (p0) circle (0.5);

\node at (0,-0.25) {$p_0$};

\foreach \i in {1,...,6}{
	\node at (-60+60*\i:0.5) (p\i) {};
	\draw (p\i) circle (1.5);
	\node at (-60+60*\i:0.8) {$p_\i$};
}

\foreach \i in {7,8,...,12}{
	\node at (-60+60*\i:2) (p\i) {};
	\draw (p\i) circle (0.86);
	\node at (-60+60*\i:2.3) {$p_{\i}$};
}

\foreach \i in {13,14,...,18}{
	\node at (-35+60*\i:2) (p\i) {};
	\node at (-30+60*\i:2.4) {$p_{\i}$};
}

\foreach \i in {0,1,...,18}{
	\draw[thick, fill=white] (p\i) circle (0.06); 
}
\end{tikzpicture}
\caption{Lower bound on the competitive ratio of \NN{}. The light gray disk (of radius $1$) represents the optimal 
solution on the boundary of which 
points~$p_7,\ldots,p_{18}$ are placed. 
Points~$p_1,\ldots,p_6$ are placed on the 
boundary of a disk of radius~$\varepsilon$ (in dark gray). 
The algorithm \NN 
is forced to place one first disk of radius~$\varepsilon$ around~$p_0$, 
then six disks of radius~$1-\varepsilon$ around~$p_1,\ldots,p_6$. 
And finally six disks of radius roughly 0.5 around~$p_7,\ldots,p_{12}$. }
\label{fig:lower_bound_NN}
\end{center}
\end{figure}
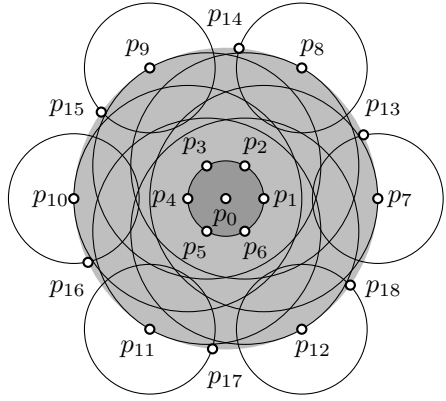

Thus, we have that for any~$\varepsilon>0$, there is 
an instance on which a solution of cost 1 exists, whereas \NN 
is forced to pay~$\varepsilon^\alpha + 6 (1-\varepsilon)^\alpha 
	+ 6 (2\sin (\pi/12) -\delta)^\alpha$, 
where~$\delta > 0$ tends to~$0$ as~$\varepsilon$ tends to~$0$. 
We can therefore conclude that \NN has to pay at least~$6 (1 + (2 \sin(\pi/12))^\alpha)  
	= 6 (1+ (2 \frac{\sqrt{6}-\sqrt{2}}{4})^\alpha)  
	= 6 (1 + (\frac{\sqrt{6}-\sqrt{2}}{2})^\alpha)   
	\approx 6 (1+ 0.52^\alpha)  $, 
whereas~$\opt \leqslant 1$.  
	
We can then scale this construction and thus, there 
is no constant~$a$ such that~$\cost(\NN) \leq c \cdot \OPT + a$ 
for $c < 6 (1 + (\frac{\sqrt{6}-\sqrt{2}}{2})^\alpha)$.
\end{proof}

\subsection{Bounds on the competitive ratio of \NN and 2-\NN when~$\alpha=2$}
Above we proved upper bounds for \NN and \CI for~$\alpha >2$, and we gave a lower
bound for \NN for any $\alpha>1$. We now study \NN and 2-\NN 
for the case~$\alpha=2$. Unfortunately, the arguments below do not apply to \CI. 

\mypara{An upper bound on the competitive ratio of~2-\NN for~$\alpha =2$.}
Let~$\objects := p_0,p_1,\ldots,p_{n-1}$ be the input instance.
Recall that~$\nn(p_i)$ is the closest point to~$p_i$ among~$p_0,\ldots,p_{i-1}$. 
Upon insertion of a point~$p_i$, if~$p_i$ is not covered by the current set of 
balls~$B(p_{i'},r_{j-1}(p_{i'}))$ with $i'<i$, then 
2-\NN increases the range of $\nn(p_i)$ to~$2\cdot\dist(p_i,\nn(p_i))$, 
and otherwise it does nothing. 
Suppose that upon the insertion of some point~$p_i$, we increase 
the range of~$\nn(p_i)$. 
We now define~$D^*_i$ as the disk centered at~$p_i$ (\emph{not} at $\nn(p_i)$) and of radius~$d_i/2$, 
where~$d_i:=\dist(p_i,\nn(p_i))$.
We call~$D^*_i$ the \emph{charging disk} of~$p_i$. Note that the 
charging disk is a tool in the proof, it is not a disk used by the algorithm. 
If~2-\NN did nothing upon insertion of~$p_i$
because~$p_i$ was already covered by a disk, we define~$D_i^*:=\emptyset$. 
\begin{lemma}\label{D_disjoint_3-NN}
For every pair of charging disks~$D^*_i$ and~$D^*_j$ with~$j\neq i$, 
we have $D^*_i \cap D^*_j = \emptyset$. 
\end{lemma}
\begin{proof}
Without loss of generality we assume that $i<j$. Suppose for a contradiction
that $D^*_i \cap D^*_j\neq\emptyset$. 
Let~$p_{i'}:=\nn(p_i)$ and~$p_{j'} := \nn(p_j)$, and let~$d_i := \dist(p_i,p_{i'})$ and~$d_j := \dist(p_j,p_{j'})$. 
Since~$i'<i<j$, we have $\dist(p_j,p_{i'}) > 2d_i$, otherwise~$p_j$ lies inside the disk of~$p_{i'}$ when~$p_j$ is inserted and we would have~$D_j^* = \emptyset$. On the other hand, $d_i/2+d_j/2\geq \dist(p_i,p_j)$
because $D^*_i \cap D^*_j\neq\emptyset$. Since $d_j \leq \dist(p_i,p_j)$, which
is true because we assumed~$i<j$, this implies $d_i\geq \dist(p_i,p_j)$.
But then $\dist(p_j,p_{i'}) \leq d_i +  \dist(p_i,p_j) \leq 2d_i$, a contradiction.
\end{proof}

\begin{lemma}\label{lem:D-in-1.5}
For any points~$p_i$ and~$p_j$ with~$i<j$, let~$D_j^{\OPT}(p_i)$ be 
the disk centered at~$p_i$ after~$p_j$ is inserted in an optimal 
solution and let~$\rho_{j}(p_i)$ be its radius. Furthermore, 
let~$D_{j}^{1.5 \OPT}(p_i)$ be the disk centered at $p_i$ of 
radius $1.5 \cdot \rho_{j}(p_i)$. Then, for every point $p_k$, 
there is a point~$p_i$ such that
the charging disk $D^*_k $ is contained in~$D_{k}^{1.5 \OPT}(p_i)$.
\end{lemma}
\begin{proof}
Let $p_i$ be such that~$p_k$ is contained in~$D^{\OPT}_k(p_i)$. Upon insertion of~$p_k$, 
we create the charging disk~$D^*_k$ of 
radius~$\frac{1}{2} \dist(p_k,\nn(p_k)) \leq \frac{1}{2} \dist(p_i,p_k)$ 
centered at~$p_k$. Therefore, 
the point of~$D^*_k$ furthest from~$p_i$ is at distance at 
most~$\frac{3}{2} \dist(p_i,p_k)$. Thus~$D^*_k \subset D_k^{1.5 OPT}(p_i)$. 
\end{proof}
Using these two lemmas, we can conclude the following. 
\begin{theorem}
In $\Reals^2$ the strategy 2-\NN is 36-competitive for~$\alpha=2$.
\end{theorem}
\begin{proof}
Recall that the charging disk~$D^*_i$ has radius~$\dist(p_i,\nn(p_i))/2$. 
Thus the cost incurred by~2-\NN  upon the insertion of~$p_i$ is at 
most~$(2\cdot \dist(p_i,\nn(p_i)))^2 \leq 16\cdot\radius (D^*_i)^2$. 
By Lemma~\ref{D_disjoint_3-NN}, the disks~$D^*_i$ are pairwise disjoint. 
Let~$\mathcal{D}_{\opt}$ denote the set of disks in an optimal solution, 
and let $\OPT$ be its cost.
Then by Lemma~\ref{lem:D-in-1.5} we have
$\sum_{i=1}^{n-1} \radius(D^*_i)^2 \leq \sum_{D \in \mathcal{D}_{\opt}} ((3/2) \cdot \radius(D))^2 = \frac{9}{4} \OPT$.
Hence the total cost incurred by~2-\NN is 
bounded by $16 \cdot \sum_{i=1}^{n-1} \radius(D^*_i)^2	\leq 36 \opt$. 
\end{proof}

\mypara{Upper bound on the competitive ratio of \NN{}  for~$\alpha =2$.}
We now prove an upper bound on the competitive ratio 
of \NN using a similar strategy as for~2-\NN. The proof uses
charging disks, as above. The main difference being how the charging disks are defined. 

Suppose that \NN increases the range of $\nn(p_i)$ upon the insertion of~$p_i$.
Then the \emph{charging disk}~$D^*_i$ is the disk of radius $\gamma \cdot d_i$
that is centered on the midpoint of the segment $p_i \nn(p_i)$, 
where~$d_i :=\dist(p_i,\nn(p_i))$ and~$\gamma$ is a constant to be determined later. 
If \NN did nothing upon insertion of~$p_i$, we define~$D_i^*:=\emptyset$. 
\ifFullVersion
We now show that 
\fi
the charging disks are disjoint if we pick $\gamma$ suitably.
\begin{lemma}\label{lem:D_disjoint_NN}
Let $\gamma < \frac{3-\sqrt{7}}{4}$. 
Then for every pair~$D^*_i,D^*_{j}$ of charging disks with~$i\neq j$, we have~$D^*_i \cap D^*_{j} = \emptyset$. 
\end{lemma}
\ifFullVersion
\begin{proof}
Let $p_i$ and $p_j$ be two points with charging disks $D_i^*$ and $D_j^*$. 
Let $p_{i'}:=\nn(p_i)$ and $p_{j'}:=\nn(p_j)$. 
Let also~$D_{i'}$, respectively~$D_{i}$, be the disk of radius~$\dist(p_i,p_{i'})$ centered on~$p_{i'}$,  respectively~$p_{i}$. 
We define~$D_{j'}$ and~$D_{j}$ similarly. 
Let also $m_i$ and $m_j$ be the midpoints between $p_i$ and $p_{i'}$, 
and $p_j$ and $p_{j'}$ respectively. 
We assume without loss of generality that $\dist(p_{i'},p_i)=1 \geqslant \dist(p_{j'},p_j)$.
We distinguish two cases. 

\begin{itemize}
\item First, if $i'=j'$, then $i>j$ otherwise $p_j\in D_{i'}$ when 
$p_j$ is inserted and~$D_j^*=\emptyset$. 
Moreover, let $H$ be the halfplane 
defined by the bisector of $p_{i'}=p_{j'}$ and $p_j$ with $p_j\in H$. Then, 
$p_i \notin H$ otherwise $\nn(p_i)$ is not $p_{i'}$ but $p_j$. 
This implies that the angle between~$p_{i'}p_i$ 
and~$p_{i'}p_j$ is at least $\pi/3$ (see Figure \ref{fig:i=i'}).
\begin{figure}[h]
\begin{center}
\begin{tikzpicture}
\node at (0,0) (i) {};
\node at (1,0) (m) {};
\node at (2,0) (k') {};

\draw[fill=mygray-medium, opacity=0, fill opacity=1] (-3,-2.5) rectangle (1,2.5);
\draw[fill=white, opacity=0, fill opacity=1] (i.center) circle (2);

\node at ([xshift=-7mm, yshift=-2mm]i.center) {$p_{i'}=p_{j'}$};
\node at ([xshift=4mm, yshift=-2mm]k'.center) {$p_j$};

\draw[thick] (i.center) circle (2);
\draw[thick] ([yshift=-25mm]m.center) -- ([yshift=25mm]m.center);
\node at (2.5,2) {$H$};

\draw[thick] (i.center) -- (k'.center);
\draw[thick] (i.center) -- (-60:2);

\draw[thick] (0.3,0) arc (0:-60:0.3);
\node at (0.4,0.3) {$\pi/3$};

\foreach \i in {i,m,k'}{
	\draw[thick, fill=white] (\i.center) circle (0.05);
}
\end{tikzpicture}
\end{center}
\caption{The gray area depicts where~$p_i$ can be. Recall 
that~$\dist(p_j,p_{j'})\leqslant \dist(p_i,p_{i'})$.}
\label{fig:i=i'}
\end{figure}
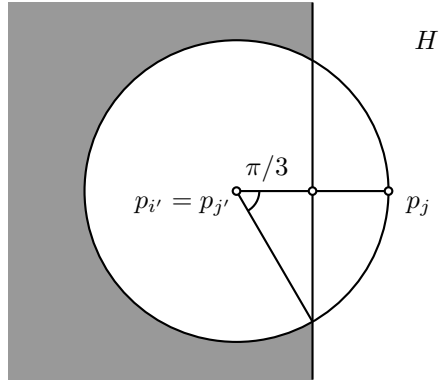

Let the two half-lines starting at $p_{i'}$ with an angle of $\pi/6$ 
with $p_{i'}p_i'$ define a wedge $w$. 
If~$\gamma$ is such that~$D_j^*$ is contained in the 
wedge $w$, the disks $D_i^*$ and $D_j^*$ are disjoint. 
That is the case when the square triangle of hypotenuse $1/2$ 
and angle $\pi/6$ has its short side at most $\gamma$. 
Using trigonometry (see Figure \ref{fig:i=i'-trig}), 
we get~$\gamma \leqslant \sin (\pi/6)/2=1/4$ 
which is always the case since~$\gamma < \frac{3-\sqrt{7}}{4}$.
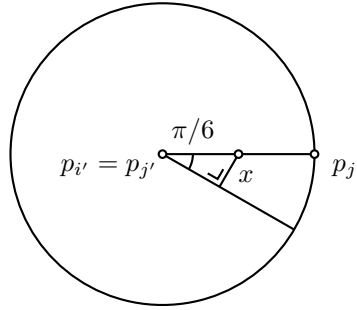
\begin{figure}[h]
\begin{center}
\begin{tikzpicture}
\node at (0,0) (i) {};
\node at (1,0) (m) {};
\node at (2,0) (k') {};

\node at ([xshift=-7mm, yshift=-2mm]i.center) {$p_{i'}=p_{j'}$};
\node at ([xshift=4mm, yshift=-2mm]k'.center) {$p_j$};

\draw[thick] (i.center) circle (2);

\draw[thick] (i.center) -- (k'.center);
\draw[thick] (i.center) -- (-30:2);
\draw[thick] (m.center) --++ (-120:0.5);
\draw[thick] ([xshift=-2mm, yshift=-2mm]m.center) --++ (-120:0.15)
			--++ (150:0.15);

\node at ([xshift=1mm, yshift=-3mm]m.center) {$x$};

\draw[thick] (0.4,0) arc (0:-30:0.4);
\node at (0.4,0.3) {$\pi/6$};

\foreach \i in {i,m,k'}{
	\draw[thick, fill=white] (\i.center) circle (0.05);
}
\end{tikzpicture}
\end{center}
\caption{If $\gamma\leqslant x = \sin (\pi/6)/2=1/4$, the disk 
centered at the midpoint 
is always contained in the wedge of angle $\pi /3$. }
\label{fig:i=i'-trig}
\end{figure}

\item We now deal with the case $i\neq i'$. Suppose for 
a contradiction that~$D_i^*$ and~$D_j^*$ intersect. 
Consider the interior of~$D_{i'}\cap D_{i}$. We 
claim that if~$p_j$ is in that region, then~$D_i^*$ 
and~$D_j^*$ do not intersect. 
Suppose~$p_j$ is in the interior of~$D_{i'}\cap D_{i}$. 
If~$j \geqslant i$, 
then $p_j\in D_{i'}$ and~$D_j^*=\emptyset$ 
which is a contradiction. 
If, on the other hand~$j < i$, when~$p_i$ is inserted, 
we have that~$\nn(p_i)$ is~$p_j$ and not~$p_{i'}$ 
since~$p_j$ is in~$D_{i}$, which is a contradiction. 
Therefore, from now on, we can assume 
that~$p_j \notin \interior (D_{i'} \cap D_i)$.
Note that this implies~$\dist(p_j,m_i)\geqslant 1/2$. 
Therefore, if~$\dist(p_j,m_j)< 1/2 - 2\gamma$, 
then we have that~$\dist(m_i,m_j) > 2\gamma$ and 
thus~$D_{i}^*$ and~$D_{j}^*$ can never intersect because 
the radius of~$D_i^*$  is~$\gamma\dist(p_{i'},p_i)=\gamma$ 
and the radius of~$D_j^*$ is~$\gamma\dist(p_{j'},p_j)\leqslant\gamma$.
Hence~$\dist(p_j,m_j)\geqslant 1/2 - 2\gamma$ which 
implies~$\dist(p_j,p_{j'})\geqslant 1 - 4 \gamma$. 

Moreover, we claim that~$p_j$ has to be in the disk~$D_{m_i}$ of 
radius~$2\gamma +1/2$ centered on~$m_i$ for the disks~$D_i^*$ 
and~$D_j^*$ to intersect. 
Suppose it is not 
the case. Then~$\dist(p_j,m_i)>2\gamma +1/2$ which implies that~
$\dist(m_i,m_j)\geqslant \dist(p_j,m_i) -  \dist(p_j,m_j) >  2\gamma$ 
and then the disks~$D_i^*$ and $D_j^*$ are disjoint. 
Figure~\ref{fig:i-neq-i'} shows the 
region~$A:= D_{m_i} \setminus \interior(D_i \cap D_{i'})$ 
where~$p_j$ has to be in order to have the disks 
intersect. 
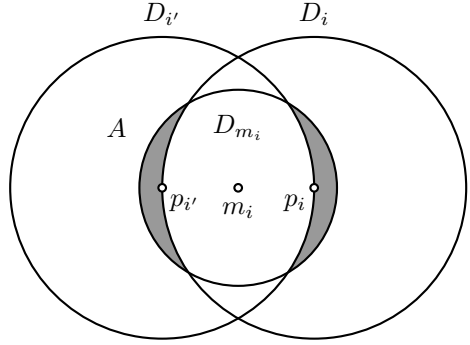
\begin{figure}[h]
\begin{center}
\begin{tikzpicture}
\node at (0,0) (i) {};
\node at (1,0) (m) {};
\node at (2,0) (k) {};

\draw[fill=mygray-medium, opacity=0, fill opacity=1] (m.center) circle (1.3);
\begin{scope}
			\clip (i) circle (2);
			\draw[fill=white, opacity=0, fill opacity=1] (k) circle (2);
\end{scope}

\node at ([xshift=3mm, yshift=-2mm]i.center) {$p_{i'}$};
\node at ([xshift=-2.5mm, yshift=-2mm]k.center) {$p_i$};
\node at ([xshift=0mm, yshift=-3mm]m.center) {$m_i$};

\draw[thick] (i.center) circle (2);
\draw[thick] (k.center) circle (2);
\draw[thick] (m.center) circle (1.3);

\node at ([xshift=-16mm, yshift=8mm]m.center) {$A$};
\node at ([xshift=0mm, yshift=8mm]m.center) {$D_{m_i}$};
\node at ([yshift=23mm]i.center) {$D_{i'}$};
\node at ([yshift=23mm]k.center) {$D_i$};

\foreach \i in {i,m,k}{
	\draw[thick, fill=white] (\i.center) circle (0.05);
}
\end{tikzpicture}
\end{center}
\caption{The region~$A:= D_{m_i} \setminus \interior(D_i \cap D_{i'})$ 
in gray depicts where~$p_j$ needs to be for the 
disks to intersect.  }
\label{fig:i-neq-i'}
\end{figure}

Using trigonometry, we compute the maximum distance between~$p_j$ 
and either~$p_{i'}$ or~$p_i$, depending on which is closer, that 
is~$\max_{p_j\in A} \min (\dist(p_{i'}, p_j), \break
	\dist(p_i, p_j))$. See Figure~\ref{fig:i-neq-i'-trig}.
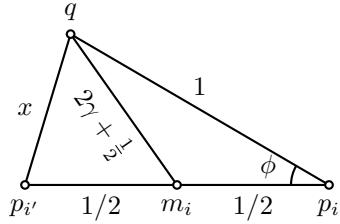
\begin{figure}[h]
\begin{center}
\begin{tikzpicture}
\node at (0,0) (i) {};
\node at (2,0) (m) {};
\node at (4,0) (k) {};
\node at (0.6,2) (k') {};

\node at ([xshift=0mm, yshift=-3mm]i.center) {$p_{i'}$};
\node at ([xshift=0mm, yshift=-3mm]k.center) {$p_i$};
\node at ([xshift=0mm, yshift=-3mm]m.center) {$m_i$};
\node at ([xshift=0mm, yshift=3mm]k'.center) {$q$};

\node at (1,-0.3) {$1/2$};
\node at (3,-0.3) {$1/2$};
\node at (2.3,1.3) {$1$};
\node at (0,1) {$x$};
\node[rotate=-55] at (1,0.8) {$2\gamma +\frac{1}{2}$};

\draw[thick] ([xshift=-5mm]k.center) arc (180:150:0.5);
\node at ([xshift=-8mm, yshift=2.5mm]k.center) {$\phi$};

\draw[thick] (i.center) -- (k.center) -- (k'.center) -- (i.center);
\draw[thick] (k'.center) -- (m.center);

\foreach \i in {i,m,k,k'}{
	\draw[thick, fill=white] (\i.center) circle (0.05);
}
\end{tikzpicture}
\end{center}
\caption{The point~$q$ of the triangle is defined as the intersection 
of~$D_{m_i}$ and~$D_{i'}$. We want to compute~$x$. }
\label{fig:i-neq-i'-trig}
\end{figure}

Using the law of cosines, we have 
\begin{align*}
& & \left(2\gamma +\frac{1}{2}\right)^2 &= \left(\frac{1}{2}\right)^2 
		+1^2 - 2 \cdot 1\cdot \frac{1}{2} \cdot \cos(\phi) \\
&\mbox{which is equivalent to} & 4\gamma^2 +2 \gamma &= 1- \cos(\phi) \\
&\mbox{and} & x^2 &= 1^2 
		+1^2 - 2\cdot 1\cdot 1 \cdot \cos(\phi) \\
&\mbox{which is equivalent to} & \frac{1}{2}x^2 &= 1- \cos(\phi). 
\end{align*}
Combining the two equations, 
we get~$x= 2 \sqrt{\gamma(2\gamma+1)} $. 
Since~$\gamma < \frac{3-\sqrt{7}}{4}$, 
then~$2 \sqrt{\gamma(2\gamma+1)} < 1 - 4\gamma$. 
Thus~$x < \dist(p_j,p_{j'})$. 
Consequently, we have that~$p_j$ 
is closer to either $p_{i'}$ or $p_i$ than it is 
to~$p_{j'}$. We show that 
both these options lead to a contradiction. 

Let us first assume that~$p_j$ is in the right crescent, 
so~$p_j$ is closer to~$p_i$ than~$p_{j'}$. 
If~$i<j$, then~$\dist(p_j,p_i) \leqslant x < \dist(p_j,p_{j'})$ which is a contradiction. 
Otherwise, if~$j<i$ then~$p_j$ is closer to~$p_i$ than~$p_{i'}$ 
when~$p_i$ is inserted, which is also a contradiction.  

Then, assume~$p_j$ is in the left crescent, 
so~$p_j$ is closer to~$p_{i'}$ than~$p_{j'}$. 
That implies~$j\leqslant i'$ otherwise~$p_{i'}$ is closer to~$p_j$ than~$p_{j'}$ 
when~$p_j$ is inserted, which is a contradiction. 
Note that if~$j = i'$, then there are only three points, but 
all the following arguments hold the same way. 
We hence have that~$j'<j\leqslant i'<i$. If~$p_{i}\in D_{j'}$ 
it implies that~$D_i^*=\emptyset$ leading to a 
contradiction. 
Therefore we have that~$\dist(p_i, p_{j'})> \dist(p_j,p_{j'}) \geqslant 1-4\gamma$. 
We now compute~$\dist(p_i, m_j)$ using Apollonius's 
theorem on the triangle~$\Delta p_ip_jp_{j'}$ 
and the median~$p_i m_j$: 
\begin{align*}
\dist(p_i,p_j)^2 + \dist(p_i&,p_{j'})^2 
	= 2(\dist(p_i,m_j)^2 + \dist(p_j,m_j)^2) \\[6pt]
\mbox{hence} \qquad 
2\dist(p_i,m_j)^2 
	&= \dist(p_i,p_j)^2 + \dist(p_i,p_{j'})^2 
	- 2 \dist(p_j,m_j)^2 \\
	&> 1 + (1-4\gamma)^2 - 2 \left( \frac{1}{2} \right)^2 \\
	&= 16\gamma^2 - 8\gamma + \frac{3}{2} 
\end{align*}
whose infimum in~$\left( 0,\frac{3-\sqrt{7}}{4} \right)$ 
is obtained when~$\gamma=\frac{3-\sqrt{7}}{4}$. 
We therefore have that~$2\dist(p_i,m_j)^2 > \frac{23-8\sqrt{7}}{2}$ 
and so~$\dist(p_i,m_j) > \frac{4-\sqrt{7}}{2}$. 
This implies that~$\dist(m_i,m_j)
	\geqslant \dist(p_i,m_j) - \dist(p_i,m_i)
	> \frac{3-\sqrt{7}}{2}
	= 2\gamma$.  
Thus~$D_i^* \cap D_j^* =\emptyset$ which is a contradiction. 
\end{itemize}
This concludes the lemma. 
\end{proof} 
\fi

We also need the following lemma, whose proof is similar to that of Lemma~\ref{lem:D-in-1.5}. 
\begin{lemma}\label{lem:D_1.5+alpha}
For any points~$p_i$ and~$p_j$ with~$i<j$, let~$D_j^{\OPT}(p_i)$ be the disk 
centered at~$p_i$ of radius~$\rho_{j}(p_i)$ 
after~$p_j$ is inserted 
in an optimal 
solution and let~$D_{j}^{(1.5+\gamma) \OPT}(p_i)$ be the disk centered at~$p_i$ of 
radius~$(1.5+\gamma) \rho_{j}(p_i)$. Then, for every point~$p_k$, 
there is a point~$p_i$ such that
the disk~$D^*_k $ is contained in 
the disk~$D_{k}^{(1.5+\gamma) \OPT}(p_i)$.
\end{lemma}
Putting everything together we obtain the following theorem.
\begin{theorem}
In $\Reals^2$ the strategy \NN is 322-competitive for~$\alpha=2$.
\end{theorem}
\begin{proof}
Recall that $\radius(D^*_i) = \gamma \cdot\dist(p_i,nn(p_i))$. 
Thus the cost incurred by~\NN upon the insertion of~$p_i$ is at 
most~$\dist(p_i,nn(p_i))^2 \leq ((1/\gamma) \cdot\radius(D^*_i))^2$.
By Lemma~\ref{lem:D_disjoint_NN}, the disks~$D^*_i$ are pairwise disjoint. 
If~$\mathcal{D}_{\opt}$ denotes the set of disks used in an optimal solution, 
then by Lemma~\ref{lem:D_1.5+alpha} we have
$\sum_{i=1}^{n-1} \rho(D^*_i)^2 \leq \sum_{D \in \mathcal{D}_{\opt}} ((1.5+\gamma) \cdot \rho(D))^2 = (1.5+\gamma)^2 \OPT$,
where~$\OPT$ is the cost of an optimal solution.
Hence the total cost incurred by~\NN is 
at most~$\frac{1}{\gamma^2} \sum_{i=1}^{n-1} \rho(D^*_i)^2
	= \frac{1}{\gamma^2} \cdot (1.5+\gamma)^2 \opt$. 
Since this holds for any value of~$\gamma < \frac{3-\sqrt{7}}{4}$, we can 
conclude that the cost incurred by~\NN is 
at most~$\frac{4^2}{(3-\sqrt{7})^2} \cdot (1.5+\frac{3-\sqrt{7}}{4})^2 \opt
	= (163 + 60 \sqrt{7}) \opt
	< 322 \opt $.
\end{proof}
 \section{Online range-assignment in general metric spaces}
\label{se:metric}
In this section we consider the problem in general metric spaces. 
\ifFullVersion
We
\fi
also consider the offline variant of the problem in the next section; here we focus on the online variant,
for which we give an $O(\log n)$-competitive algorithm. 
The key insight to our algorithms is to formulate the problem as a set-cover problem
and apply linear-programming techniques. As we will see later, applying the online set cover algorithm of Alon~\etal~\cite{AlonAABN09} yields a competitive ratio much worse than $O(\log n)$, so we need to exploit structural properties of the particular set cover instances arising from our problem.

\subsection{A set cover formulation and its LP}
\label{se:gen:LP}
Let $\radii$ be the set of distances between pairs of points. Observe that we can restrict ourselves 
without loss of generality to only using ranges from~$\radii$. This allows us to formulate the problem 
in terms of set cover: The elements are the points $p_0, p_1, \ldots, p_{n-1}$, with $p_0$ being the source point, 
and for each $0 \leq i \leq n-2$ 
and $r \in \radii$ there is a set $S_{i,r} := \{p_j : j > i \mbox{ and } \dist(p_i, p_j) \leq r\}$ with cost $r^\alpha$.
(Note that $S_{i,r}$ is the set of points arriving after $p_i$ that are within range~$r$ of $p_i$).
In the following, we abuse notation and also write $j \in S_{i,r}$ for points $p_j \in S_{i,r}$. We also say that $S_{i,r}$ is \emph{centered} at $p_i$.

Observe that a feasible range assignment corresponds to a feasible set cover. A set cover 
is \emph{minimally feasible} if removing any set from it causes an element to be uncovered. 
Since a minimally feasible set cover picks at most one set $S_{i,r}$ for each $i$, 
it corresponds to a feasible range assignment. (Note that applying the online set cover algorithm of Alon~\etal~\cite{AlonAABN09} only gives a competitive ratio of $O(\log^2 n/ \log \log n)$ as our set cover instance has $n-1$ elements and $|\radii|(n-1)$ sets.)

We can now formulate our problem as an integer linear program.
To this end we introduce, for each range $r\in\radii$ and each point~$p_i$ a
variable $x_{i,r}$, where $x_{i,r}=1$ indicates we choose the set $S_{i,r}$ (or, in other words,
that we assign range~$r$ to~$p_i$) and $x_{i,r}=0$ indicates we do not choose~$S_{i,r}$. Allowing the $x_{i,r}$ to take fractional values gives us the following relaxed LP. 
\begin{equation}
\boxed{
  \label{lp:P}
\begin{aligned}
  \mbox{Minimize} \quad            
  & \sum_{0 \leq i \leq n-2} \sum_{r \in \radii} x_{i,r}\cdot r^\alpha\\
  \mbox{Subject to}\quad  
  & \sum_{i, r : j \in S_{i,r}} x_{i,r} \geq 1 &\quad \mbox{ for all } 1 \leq j \leq n-1\\
  & x_{i,r} \geq 0 &\quad \mbox{ for all } (i,r) \mbox{ with } 0 \leq i \leq n -2 \mbox{ and } r \in \radii
\end{aligned} 
}
\end{equation}
The dual LP corresponding to the LP above is as follows.
\begin{equation}
\boxed{
  \label{lp:D}
\begin{aligned}
  \mbox{Maximize} \quad            
  & \sum_{1 \leq j \leq n} y_j \\
  \mbox{Subject to}\quad  
  & \sum_{j \in S_{i,r}} y_j \leq r^\alpha &\quad \mbox{ for all } (i,r) \mbox{ with } 0 \leq i \leq n-2 \mbox{ and } r \in \radii \\
    & y_j \geq 0 &\quad \mbox{ for all } 1 \leq j \leq n - 1
\end{aligned} 
}
\end{equation}
We say that the set $S_{i,r}$ is \emph{tight} if the corresponding dual constraint is tight, 
that is, if $\sum_{j \in S_{i,r}} y_j = r^\alpha$.

\subsection{The online algorithm and its analysis}
\label{se:gen:onub}
Recall that in the online version, we are given the source $p_0$ and then the points $p_1, \ldots, p_{n-1}$ 
arrive one-by-one. When a point $p_i$ arrives, its distances to previous points and the source are revealed.

\subparagraph{The algorithm.} Let $\gamma > 1$ be a constant that we will set later. The basic idea of the algorithm is that when a point $p_i$ arrives, we will raise its associated dual variable $y_i$ until some set $S_{j,r}$ containing $p_i$ is tight and then update the range of point $p_j$ to be $r_i(p_j) := \gamma \max\{r : \sum_{k \in S_{j,r} : k \leq i} y_k = r^\alpha\}$. In other words, the range of $p_j$ becomes $\gamma$ times the largest radius of the tight sets centered at $p_j$.

Here is a more precise description of the algorithm. When $p_i$ arrives, we initialize its dual variable~$y_i := 0$. If $p_i \in S_{j,r}$ for some $j < i$ and range $r$ with $\sum_{k \in S_{j,r} : k \leq i} y_k = r^\alpha$, then we set $r_i(p_j) := \gamma \max\{r : \sum_{k \in S_{j,r} : k \leq i} y_k = r^\alpha\}$ for one such~$j$. (It can happen that some $S_{j,r}$ is tight but that $r_{i-1}(p_j)$ is still smaller than $r$, because when multiple sets become tight at the same time, we only increase the range of one point.) Otherwise, we increase $y_i$ until for some $j < i$ and range $r$ we have $\sum_{k \in S_{j,r} : k \leq i} y_k = r^\alpha$; we then set $p_j$'s new range to $r_{i}(p_j) := \gamma r$ for one such~$j$. In both cases, we only set $p_j$'s range, the other ranges remain unchanged. Note that in the event that multiple sets centered at different points become tight simultaneously, we only update the range of one of them.

\subparagraph{Analysis.} 
We begin our analysis of the algorithm by showing the feasibility of the constructed dual solution $y$ and the corresponding range assignment. For each point $p_i$, the algorithm stops raising $y_i$ once some set $S_{j,r}$ containing $p_i$ is tight and then updates $p_j$'s radius to be $\gamma r > r$. This guarantees that no dual constraint is violated and that $p_i$ is covered by $p_j$.

Next we analyze the cost of this algorithm. We use the shorthand $r_j$ for the final range $r_{n-1}(p_j)$ 
of the point~$p_j$. First, we argue that it suffices to bound 
the cost of the points whose ranges are large enough. 
Let $H = \{0 \leq i \leq n-2 : r_i \geq \max_{0 \leq j \leq n-2} r_j/n\}$. Then, the cost of the algorithm is
\[
\sum_i r_i^\alpha = \sum_{i \in H} r_i^\alpha + \sum_{i \notin H} r_i^\alpha 
                \leq \sum_{i \in H} r_i^\alpha + n (\max_j r_j/n)^\alpha 
                \leq (1 + 1/n^{\alpha-1}) \sum_{i \in H} r_i^\alpha 
                \leq 2\sum_{i \in H} r_i^\alpha,
\]
where the second last inequality is because $\sum_{i \in H} r_i^\alpha \geq \max_j r_j^\alpha$ and the last is because $\alpha > 1$. 
In the remainder of this section 
we will show that 
\begin{equation}
\sum_{i \in H} r_i^\alpha \leq O(\log n) \cdot \sum_{1 \leq j \leq n-1} y_j.  \label{eq:gen-online-dual}
\end{equation}
The theorem then follows from the Weak Duality Theorem of Linear Programming  
which states that value of any feasible solution to the primal (minimization) problem is always 
greater than or equal to the value of any feasible solution to its associated dual problem. 

For $0 \leq i \leq n-2$, our algorithm sets the final range~$r_i$ of point $p_i$ such that 
$r_i = \gamma r$ for some $r \in \radii$ such that $\sum_{k \in S_{i,r}} y_k = r^\alpha$. 
Thus, we get 
\[
\left(\frac{r_i}{\gamma}\right)^\alpha = \sum_{j \in S_{i, r_i/\gamma}} y_j,
\]
and so
\[
\sum_{i \in H} r_i^\alpha 
    = \sum_{i \in H} \gamma^\alpha \left(\sum_{j \in S_{i, r_i/\gamma}} y_j\right) 
    = \gamma^\alpha \sum_{1 \leq j \leq n-1}  y_j \cdot \left| \{i \in H : j \in S_{i, r_i/\gamma}\} \right|,
\]
where the last equality follows by interchanging the sums. Thus, to prove Inequality~\eqref{eq:gen-online-dual}
it suffices to prove the following lemma.

\begin{lemma}
  For every $1 \leq j \leq n-1$ and any fixed $\gamma > 3$, we have $|\{i \in H : j \in S_{i, r_i/\gamma}\}| = O(\gamma^\alpha\log n)$.
\end{lemma}

\begin{proof}
  Define $H_j = \{i \in H : j \in S_{i, r_i/\gamma}\}$. We will show that 
  for every $i, i' \in H_j$, either $r_i > \frac{\gamma - 1}{2} r_{i'}$ or $r_{i'} > \frac{\gamma - 1}{2} r_i$.
  This implies that the $t$-th smallest range (among the points in $H_j$)
  is at least $((\gamma-1)/2)^t$ times the smallest range (among those points). 
  Since $\frac{\max_{i \in H_j} r_i}{\min_{i \in H_j} r_i} \leq n$, this means that
  $|H_j| = O(\log_{(\gamma-1)/2} n) = O(\log n)$. 

  Suppose $i, i' \in H_j$. Let $p_k$ be the last-arriving point that causes our 
  algorithm to update $r_i$, and $p_{k'}$ be the last-arriving point that causes 
  our algorithm to update $r_{i'}$. Since the arrival of any point causes at most one point's range to be updated, we have that $p_k \neq p_{k'}$. Suppose that $p_k$ arrived before $p_{k'}$. 
  By construction of $r_{i'}$, we have $\dist(p_{k'},p_{i'}) = r_{i'}/\gamma$. 
  Moreover, since $i, i' \in H_j$, we have $\dist(p_i, p_j) \leq r_i/\gamma$ 
  and $\dist(p_{i'}, p_j) \leq r_{i'}/\gamma$. Therefore, by the triangle inequality,
  \[
  \dist(p_i, p_{k'}) 
        \leq \dist(p_i, p_j) + \dist(p_j,p_{i'}) + \dist(p_{i'}, p_{k'}) 
        \leq 2\frac{r_{i'}}{\gamma} + \frac{r_i}{\gamma}.
  \]
  Since $p_{k}$ arrived before $p_{k'}$ and $p_{k'}$ caused our algorithm to 
  update $r_{i'}$, the point $p_{k'}$ must have been uncovered when it arrived, and so $\dist(p_i, p_{k'}) > r_i$. 
  Therefore, we get
  \[
  r_i < \dist(p_i,p_{k'}) \leq 2\frac{r_{i'}}{\gamma} + \frac{r_i}{\gamma}
  \]
  and so $r_{i'} > \frac{\gamma - 1}{2} r_i$
  as desired. In the case that $p_{k'}$ arrived before $p_k$, a similar 
  argument yields $r_i > \frac{\gamma - 1}{2} r_{i'}$.
\end{proof}
By setting $\gamma = 4$ we obtain the following theorem.
\begin{theorem}
  \label{thm:gen:online-alg}
  For any power-distance gradient~$\alpha >1$, there is a $O(4^\alpha\log n)$-competitive algorithm 
  for the online range assignment problem in general metric spaces.
\end{theorem}
 \ifFullVersion
\section{On offline algorithm for general metric spaces}
\label{se:gen:off}
In the offline setting, we are given the entire sequence of points $p_0, \ldots, p_n$ 
in advance and the goal is to assign ranges $r_0, \ldots, r_{n-1}$ to the points $p_0, \ldots, p_{n-1}$ 
so that for every $1 \leq i \leq n$, there exists $j < i$ such that $\dist(p_i, p_j) \leq r_j$. 
We can formulate the problem in this way because we know all points beforehand, and we are
interested in the cost of the final assignment. Thus we may immediately assign each point
its final range, and we need not specify a separate range for every point at each time step.
The stated condition on the assignment ensures that after inserting each $p_j$,
we have a broadcast tree on $p_0,\ldots,p_j$. Thus we require the algorithm to
be what Boyar~\etal~\cite{befkl-ods-19} call an \emph{incremental algorithm}:
namely an algorithm that maintains a feasible solution at any time
(even though, unlike an online algorithm) it may know the future).
We emphasise that this is different from the \emph{static} broadcast range assignment problem studied previously. 
To avoid confusion with the usual offline broadcast range assignment problem, we call this the \emph{Priority Broadcast Range Assignment} problem.\footnote{The priority of a point is its position in the sequence, the lower the position, the higher its priority. Each point can only be covered by a point with a higher priority.}
Below we give a $5^\alpha$-approximation algorithm for the offline version of the problem,
based on the LP formulated in Section~\ref{se:gen:LP}.
\medskip

The basic idea of the approximation algorithm is as follows. We start with a maximally feasible dual solution $y$, i.e. increasing any $y_j$ would violate some dual constraint. Since $y$ is maximally feasible, for every $j$, there exists a set $S_{i,r}$ containing $p_j$ that is tight. Thus, the tight sets form a feasible set cover. Let $\sol$ be subset of the tight sets that is a minimally feasible set cover. As observed above, for every $i$, there is at most one set $S_{i,r} \in \sol$. Thus, $\sol$ corresponds to a feasible range assignment. Let $r_i$ be the radius assigned to $p_i$.

We now modify the range assignment $r$ to get a range assignment $r'$ so that
\[
\sum_{0 \leq i \leq n-1} {r'}_i^2 \leq 5^\alpha \sum_{1 \leq j \leq n} y_j.
\]
Since $y$ is a feasible dual solution $y$, weak LP duality implies that $r'$ is a $5^\alpha$-approximation.

Say that $i$ \emph{conflicts} with $j$ if there exists  a point $p_k \in S_{j, r_j} \cap S_{i, r_i}$ such that $y_k > 0$. Order the indices in decreasing order of $r_i$, breaking ties arbitrarily, and denote by $i \prec j$ if $i$ comes before $j$ in this ordering. We use the following algorithm to construct $r'$.

\begin{algorithm}
\caption{Obtaining an approximate solution from $\sol$}
\begin{algorithmic}[1]
  \label{alg:gen-offline}
  \STATE Initialize $I \leftarrow \emptyset$
  \FOR {$i$ in order according to $\prec$}
  \IF {$i$ does not conflict with any $j \in I$}
  \STATE Define $C_i = \{i\} \cup \{j \succeq i: \mbox{$j$ conflicts with $i$ but not with $I$}\}$
  \STATE Add $i$ to $I$
\ENDIF
  \ENDFOR
  \FOR {$i \in I$}
  \STATE Let $p_{i'}$ be the earliest point in $C_i$ 
  \STATE Assign radius $r'_{i'} \leftarrow 5r_i$ and radius $r'_j \leftarrow 0$ for $j \in C_i \setminus \{i'\}$
  \ENDFOR
\end{algorithmic}
\end{algorithm}

\subparagraph{Analysis}
We begin by proving that $r'$ is a feasible range assignment.

\begin{lemma}
  \label{lem:feas}
  For each $j > 0$, there exists $i < j$ such that $\dist(p_j, p_i) \leq r'_i$.
\end{lemma}

\begin{proof}
  Note that the sets $\{C_i\}_{i \in I}$ partition $\{1, \ldots, n\}$. Consider some set $C_i$ and let $p_{i}'$ be the earliest point in $C_i$. It suffices to prove that $S_{i', 5r_i} \supseteq \cup_{j \in C_i} S_{j,r_j}$. To see this, first observe that since $p_{i'}$ is the earliest point in $C_i$ it can potentially cover all the points covered by any other point $p_k$ for $j \in C_i$, i.e. $S_{i', r} \supseteq \cup_{j \in C_i} S_{j,r_j}$ for large enough $r$. Next, we show that $r = 5r_i$ suffices. Since $i$ conflicts with every $k \in C_i$ and $r_i = \max_{j \in C_i} r_j$, we have that every point in $\cup_{j \in C_i} S_{j,r_j}$ is within distance $3r_i$ of $p_i$ and that $\dist(p_{i'}, p_i) \leq 2r_i$. Thus, we get that $S_{i', 5r_i} \supseteq \cup_{j \in C_i} S_{j,r_j}$, as desired.
\end{proof}

\begin{lemma}
  \label{lem:cost}
  $\sum_i r'^\alpha_i \leq 5^\alpha\sum_{j > 0} y_j$.
\end{lemma}

\begin{proof}
  We have $\sum_i r'^\alpha_i = \sum_{i \in I} 5^\alpha r^\alpha_i$.  Since the sets in $\sol$ are tight, we have  
  \begin{align*}
      \sum_{i \in I} r^\alpha_i = \sum_{i \in I} \sum_{j \in S_{i,r_i}} y_j
    = \sum_j |\{i \in I : j \in S_{i,r_i}\}| y_j
  \leq \sum_j y_j
  \end{align*}
where the last inequality follows from the fact that $I$ is conflict-free.
\end{proof}

Thus, we get the following theorem.

\begin{theorem}
  There is a $5^\alpha$-approximation algorithm for the Priority Range Assignment problem in general metric spaces for any $\alpha > 1$. 
\end{theorem} \fi
\section{Concluding Remarks}
We introduced the online version of the broadcast range-assignment problem, and we analyzed
the competitive ratio of two natural algorithm, \NN and \CI,
in $\Reals^1$ and $\Reals^2$ as a function of the power-distance gradient~$\alpha$.
While \NN is $O(1)$-competitive in $\Reals^2$ and for $\alpha=2$ the best competitive ratio we can 
prove is quite large, namely~322. The variant 2-\NN has a better ratio,
namely~36, but this is still large. We conjecture that the actual competitive ratio
of \NN is actually much closer to the lower bound we proved, which is~7.61.
We also conjecture that \CI has a constant (and small) competitive ratio in $\Reals^2$.
Another approach to getting better competitive ratios might be to develop more sophisticated algorithms.
For the general (metric-space) version of the problem, the main question is whether
an algorithm with constant competitive ratio is possible.

While the requirement that we cannot decrease the range of any point in the online
setting is perhaps not necessary in practice, our algorithms have the additional
benefit that they modify the range of at most one point. Thus it can also be seen as
the first step in studying a more general version, where we are allowed to
modify (increase or decrease) the range of, say, two points. In general, it is
interesting to study trade-offs between the number of modifications and
the competitive ratio. Studying deletions is then also of interest. 
\section*{Acknowledgement}
We thank two anonymous referees for their comments on a previous version of this paper. In particular, we thank them for suggesting to consider 2-\NN (in our previous version we analyzed 3-\NN) and for the proof of Lemma~\ref{D_disjoint_3-NN}.

\bibliographystyle{plainurl}
\bibliography{online-range-assignment}{}

\end{document}